\documentclass[11pt]{article}

\usepackage[sort]{cite}       
\usepackage{amssymb}          
\usepackage{latexsym}         
\usepackage{amsmath}          
\usepackage{amscd}            
\usepackage[latin1]{inputenc} 
\usepackage{eucal}            
\usepackage{bbm}              
\usepackage{theorem}          
\usepackage{stmaryrd}         
\usepackage{mathrsfs}         
\usepackage{paralist}         
\usepackage[all,dvips]{xy}    
\usepackage{graphicx}	      
\title{Deformation Quantization of a Certain Type of Open Systems}

\author{\textbf{Florian Becher}\thanks{Florian.Becher@physik.uni-freiburg.de},
  \addtocounter{footnote}{1}
  \textbf{Nikolai Neumaier}\thanks{Nikolai.Neumaier@physik.uni-freiburg.de},
  \textbf{Stefan Waldmann}\thanks{Stefan.Waldmann@physik.uni-freiburg.de},
  \\[0.1cm]
  Fakult{\"a}t f{\"u}r Mathematik und Physik\\
  Albert-Ludwigs-Universit{\"a}t Freiburg\\
  Physikalisches Institut\\
  Hermann Herder Stra{\ss}e 3\\
  D 79104 Freiburg\\
  Germany
  }

\date{September 2009}

\renewcommand{\mathbb}[1]{\mathbbm{#1}} 

\allowdisplaybreaks

\numberwithin{equation}{section}

\newcounter{comment}

\theoremheaderfont{\normalfont\bfseries}
\theorembodyfont{\itshape}
\newtheorem{lemma}{Lemma}[section]
\newtheorem{proposition}[lemma]{Proposition}
\newtheorem{theorem}[lemma]{Theorem}
\newtheorem{corollary}[lemma]{Corollary}

\theorembodyfont{\rmfamily}

\newtheorem{definition}[lemma]{Definition}
\newtheorem{example}[lemma]{Example}
\newtheorem{remark}[lemma]{Remark}

\newenvironment{proof}[1][{}]{\par\noindent\textsc{Proof{#1}: }}{\hspace*{\fill}$\blacksquare$\smallskip\noindent\par}

\allowdisplaybreaks

\newcommand{\cc}[1]      {\overline{{#1}}}              
\newcommand{\id}         {\operatorname{\mathsf{id}}}

\newcommand{\supp}       {\operatorname{\mathrm{supp}}}  
\newcommand{\Pol}        {\operatorname{\mathrm{Pol}}}   
     
\newcommand{\tr}         {\operatorname{\mathsf{tr}}}

\newcommand{\ring}[1]    {\mathsf{#1}}

\newcommand{\I}          {\mathrm{i}}

\newcommand{\D}          {\operatorname{\mathrm{d}}}

\newcommand{\pr}         {\mathrm{pr}}
\newcommand{\tensor}     {\mathbin{\widehat{\otimes}}\,}

\newcommand{\cinf}[1]    {C^\infty({#1})}	
\newcommand{\ckomp}[1]   {C^\infty_0({#1})}	
\newcommand{\gaminf}[1]  {\Gamma^\infty({#1})}   
\newcommand{\sze}[2]     {\mathcal{#1}_{#2}}      
\newcommand{\C}          {\mathbb{C}}            
\newcommand{\R}          {\mathbb{R}}            
\newcommand{\rk}[1]      {\left({#1}\right)}     
\newcommand{\gk}[1]      {\left\{{#1}\right\}}   
\newcommand{\ek}[1]      {\left[{#1}\right]}     
\newcommand{\Exp}        {\operatorname{\mathsf{Exp}}}   
\newcommand{\PresSquares}[1]  {\mathcal{#1}}   

\newcommand{\Sys}          {\mathrm{S}}
\newcommand{\Bad}          {\mathrm{B}}
\newcommand{\System}       {{\scriptscriptstyle\mathrm{S}}}
\newcommand{\Bath}         {{\scriptscriptstyle\mathrm{B}}}
\newcommand{\Interaction}  {{\scriptscriptstyle\mathrm{I}}}
\newcommand{\starSys}      {\mathbin{\star_{\System}}}
\newcommand{\starBad}      {\mathbin{\star_{\Bath}}}

\newcommand{\qs}         {{q_{\System}^{\phantom{2}}}}
\newcommand{\ps}         {{p_{\System}^{\phantom{2}}}}
\newcommand{\qb}         {{q_{\Bath}^{\phantom{2}}}}
\newcommand{\pb}         {{p_{\Bath}^{\phantom{2}}}}
\newcommand{\xs}         {{x_{\System}^{\phantom{2}}}}
\newcommand{\xb}         {{x_{\Bath}^{\phantom{2}}}}
\newcommand{\qen}        {q_1}
\newcommand{\pen}        {p_1}
\newcommand{\qzn}        {q_2}
\newcommand{\pzn}        {p_2}
\newcommand{\qin}[1]     {q_{#1}}
\newcommand{\pin}[1]     {p_{#1}}

\newcommand{\z}[1]       {{z_{#1}}}
\newcommand{\zb}[1]      {{\overline{z}_{#1}}}

\newcommand{\wick}       {{\scriptscriptstyle \mathrm{Wick}}}
\newcommand{\weyl}       {{\scriptscriptstyle \mathrm{Weyl}}}
\newcommand{\kms}        {{\scriptscriptstyle \mathrm{KMS}}}

\newcommand{\dx}[1]{\,\mathrm{d}{#1}}
\newcommand{\dbath}{\dx{\qb}\dx{\pb}}

\begin{document}

\maketitle

\begin{abstract}
    We give an approach to open quantum systems based on formal
    deformation quantization. It is shown that classical open systems
    of a certain type can be systematically quantized into quantum
    open systems preserving the complete positivity of the open time
    evolution. The usual example of linearly coupled harmonic
    oscillators is discussed.
\end{abstract}

\tableofcontents

\newpage

%
%

\section{Introduction}
\label{sec:Intro}

Attempts at the quantization of open systems, especially dissipative
systems, have been made for quite some time. Examples can, among many
others, be found in \cite{brittin:1950a, dekker:1975a, razavy:1977a}.
In particular, some approaches to the deformation quantization of
genuinely dissipative systems have been conducted, see
\cite{dito.turrubiates:2006a, dito.leandre:2007a}. So far, it seems
that no successful attempt has been made at a mathematically
consistent systematic quantization of open systems originating from
coupled systems.

We chose the framework of deformation quantization.  The central
object of deformation quantization \cite{bayen.et.al:1978a} is the
algebra of observables. States are regarded as a derived concept in
the sense of normalized positive linear functionals on the algebra of
observables in the classical as well as in the quantum case.  The star
products used to deform the classical algebra of observables in this
process are meant to be Hermitian star products. The existence of such
star products on the smooth functions of Poisson manifolds has been
proven by \cite{kontsevich:2003a}. For the special case of symplectic
manifolds the existence has been proven earlier by
\cite{dewilde.lecomte:1983b, fedosov:1986a,
  omori.maeda.yoshioka:1991a}, see also the textbook
\cite{waldmann:2007a} for additional references.

In the manner of speaking of \cite{breuer.petruccione:2003a}, we get
an open system (classical and quantum mechanical) by constructing a
microscopic model and non-selectively integrating the degrees of
freedom of the environment.

As a first step, we give a consistent and general definition of what a
classical and quantum open Hamiltonian system in the sense of
deformation quantization should be relying on the notion of completely
positive evolutions in both cases. As main result we prove that every
classical open Hamiltonian system can be deformation quantized
preserving complete positivity of the evolution map. A by-product of
independent interest is the result that for every Hermitian star
product on a Poisson manifold there is a completely positive map into
the undeformed algebra of formal series of smooth functions deforming
the identity map. Our general formalism is exemplified for two coupled
harmonic oscillators.

This article is organized in the following way: In
Section~\ref{sec:OpenDynamicalSystems} a notion of classical open
dynamical systems in general and the notion of a classical open
Hamiltonian system used for deformation quantization in particular are
defined.  In
Section~\ref{sec:DeformationQuantizationOpenHamiltonianSystems} we
will briefly recall the notions of a Hermitian star product and the
quantum time evolution with regard to a Hermitian star product. We
prove in Theorem~\ref{theorem:ExistPosDef} that for every Hermitian
star product one has a completely positive map deforming the identity
into the formal series of smooth functions with respect to the
undeformed product. This turns out to be the main tool to show
Theorem~\ref{theorem:SomeQuantizedOpenIsCP}: every classical open
Hamiltonian system can be deformation quantized.  In
Section~\ref{section:ExampleLinearlyCoupledHarmOszisI}, as an
illustration, we give the standard example of the total time evolution
of two one-dimensional linearly coupled harmonic oscillators in the
setting of deformation
quantization. Section~\ref{section:ExampleLinearlyCoupledHarmOszisII}
contains the open time evolutions of a coupled harmonic oscillator
with respect to states on the bath oscillator corresponding to
deformed initial values and to KMS states.

%
%

\section{Classical Open Dynamical Systems}
\label{sec:OpenDynamicalSystems}

There are many ways to specify the notion of open dynamical systems. A
fairly general approach is obtained as follows: We start with a
\emph{subsystem} whose pure states are described by a smooth manifold
$\Sys$ and a \emph{bath} which is described analogously by a smooth
manifold $\Bad$. The combined total system has the Cartesian product
$\Sys \times \Bad$ as space of pure states.

An \emph{open dynamical system} is now a time evolution of (pure)
states in $\Sys \times \Bad$ where we only look at the $\Sys$-part
``ignoring'' the $\Bad$-part. More precisely, this is obtained as
follows:

On the total system we specify an ordinary dynamical system, i.e.\ a
vector field $X \in \Gamma^\infty(T(\Sys \times \Bad))$ with
flow $\Psi_t: \Sys \times \Bad \longrightarrow \Sys \times \Bad$. For
simplicity, we may assume that the flow $\Psi_t$ is complete,
otherwise we have to restrict to certain neighbourhoods in $\Sys
\times \Bad$ and finite times in the usual way. With this assumption,
$\Psi_t$ is a smooth one-parameter group of diffeomorphisms of $\Sys
\times \Bad$ with
\begin{equation}
    \label{eq:DDtPsit}
    \frac{\D}{\D t} \Psi_t = X \circ \Psi_t
    \quad
    \textrm{for all}
    \; t \in \mathbb{R}.
\end{equation}
Next we consider the canonical projection maps
\begin{equation}
    \label{eq:ProjectionMaps}
    \Sys
    \stackrel{\pr_{\System}}{\longleftarrow}
    \Sys \times \Bad 
    \stackrel{\pr_{\Bath}}{\longrightarrow}
    \Bad,
\end{equation}
which allow to decompose the tangent bundle $T(\Sys \times \Bad)$ into
\begin{equation}
    \label{eq:TangentBundleOfMtimesB}
    T(\Sys \times \Bad)
    =
    \pr_{\System}^\# T\Sys \oplus \pr_{\Bath}^\# T\Bad,
\end{equation}
where $\pr_{\System}^\# T\Sys$ and $\pr_{\Bath}^\# T\Bad$ denote the
pull-backs of the tangent bundles of $\Sys$ and $\Bad$, respectively.

Clearly, the map $\pr_{\System}$ forgets the degrees of freedom of the
bath and thus corresponds precisely to the idea that we want to ignore
the $\Bad$-part. However, for the time evolution of $\Sys$ we still
have to specify an initial condition for the bath as well. For the
moment, we restrict ourselves to pure states and allow for mixed
states later on. Thus let $\xb \in \Bad$ be a point whence we have the
embedding
\begin{equation}
    \label{eq:iotab}
    \iota_{\xb}: \Sys \ni \xs \mapsto (\xs, \xb) \in \Sys \times \Bad,
\end{equation}
which is clearly a diffeomorphism onto its image such that
$\pr_{\System} \circ \iota_{\xb} = \id_{\System}$ and $\pr_{\Bath}
\circ \iota_{\xb} = \xb$ is the constant map.
\begin{definition}[Open time evolution, pure case]
    \label{definition:OpenTimeEvolutionPure}
    For any $\xb \in \Bad$ the open time evolution $\Phi^{\xb}_t: \Sys
    \longrightarrow \Sys$ of $\Sys$ with respect to the total time
    evolution $\Psi_t$ of $\Sys \times \Bad$ and the pure state $\xb$
    of the bath is given by
    \begin{equation}
        \label{eq:PureOTE}
        \Phi_t^{\xb} = \pr_{\System} \circ \Psi_t \circ \iota_{\xb}.
    \end{equation}
\end{definition}
Of course, we have to justify this definition and examine some
consequences as well as properties of $\Phi_t^{\xb}$. First of all,
the map
\begin{equation}
    \label{eq:PhibSmooth}
    \Phi_{\phantom{t}}^\xb: \mathbb{R} \times \Sys
    \ni (t, \xs) \mapsto \Phi^{\xb}_t(\xs) \in
    \Sys
\end{equation}
is clearly smooth. However, it does not have the usual properties of
an ordinary time evolution: For a fixed time $t$ the map
$\Phi^{\xb}_t$ needs not to be a diffeomorphism, not even for small
times. We only have the following ``evolution property'' which easily
follows from the one-parameter group property of $\Psi_t$:
\begin{proposition}
    \label{proposition:OTEProperty}
    For the open time evolution we have
    \begin{equation}
        \label{eq:PhibOTE}
        \Phi^{\xb}_0 = \id_{\System}
        \quad
        \textrm{and}
        \quad
        \Phi_s^{\pr_{\Bath}(\Psi_t(\xs ,\xb))} \circ \Phi^{\xb}_t (\xs)
        = \Phi_{s+t}^{\xb} (\xs)
    \end{equation}
    for all $\xs \in \Sys$, $\xb \in \Bad$, and $s, t \in \mathbb{R}$.
\end{proposition}
\begin{example}
    \label{example:OpenRotations}
    Let $\Sys = \mathbb{R} = \Bad$ and consider the time evolution
    \begin{equation}
        \label{eq:PsiRotation}
        \Psi_t \left(\xs ,\xb \right)
        = \left( \xs \cos(\nu t)  - \xb \sin(\nu t) ,
            \xs \sin(\nu t) + \xb \cos(\nu t)\right)
    \end{equation}
    on $\Sys \times \Bad$.
    \begin{enumerate}
    \item The simplest case is obtained for $\nu \in \mathbb{R}$ being
        a non-zero constant. Then the open time evolution for $\xb \in
        \Bad$ is given by
        \begin{equation}
            \label{eq:ExampleOTE}
            \Phi_t^{\xb} (\xs) = \xs \cos(\nu t) - \xb \sin(\nu t)
        \end{equation}
        which is a diffeomorphism for small $t$ but the constant map
        for $\nu t \in \frac{\pi}{2} + \pi \mathbb{Z}$.
    \item We can also consider the case where $\nu$ is a function on
        $\Sys \times \Bad$ depending only on the radius, e.g.\ $\nu
        (\xs, \xb) = x_{\System}^2 + x_{\Bath}^2$. Then $\Psi_t$ is
        still a one-parameter group of diffeomorphisms and the flow
        lines are still concentric circles around $(0, 0)$. However,
        the points in $\Sys \times \Bad$ spin faster the further away
        from $(0, 0)$ they are.  Now the open time evolution is
        \begin{equation}
            \label{eq:ExampleOTEII}
            \Phi_t^{\xb} (\xs)=
            \xs \cos((x_{\System}^2 + x_{\Bath}^2)t) - 
            \xb \sin((x_{\System}^2 + x_{\Bath}^2)t).
        \end{equation}
        E.g.\ for $\xb = 0$ this gives $\Phi_t^0(\xs) = \xs
        \cos(x_{\System}^2 t)$ which yields
        \begin{equation}
            \label{eq:PhitNullNeverDiffeo}
            \Phi_t^0\left(\sqrt{\frac{\pi}{2t}}\right) = 0
        \end{equation}
        for all $t > 0$. Since also $\Phi_t^0 (0) = 0$ for all $t$ we
        see that $\Phi_t^0$ cannot be a diffeomorphism, even for
        arbitrarily small time $t > 0$.
  \end{enumerate}
\end{example}

From the example we conclude that the open time evolution
$\Phi_t^{\xb}$ \emph{in general} is not a solution to a probably
time-dependent differential equation on $\Sys$ alone, i.e.\ in general
there is \emph{no} time-dependent vector field $X_t \in
\Gamma^\infty(T\Sys)$ with
\begin{equation}
    \label{eq:OTEnotTDTE}
    \frac{\D}{\D t} \Phi^{\xb}_t = X_t \circ \Phi_t^{\xb}.
\end{equation}
Nevertheless, this situation of a time-dependent vector field is a
particular case of an open time evolution as the next example shows:
\begin{example}
    \label{example:TDTEisOTE}
    Let $X_t \in \Gamma^\infty(T\Sys)$ be a smooth time-dependent
    vector field on $\Sys$ and let $\cc{X} \in \Gamma^\infty(T(\Sys
    \times \mathbb{R}))$ be the corresponding time-independent vector
    field
    \begin{equation}
        \label{eq:ccX}
        \cc{X} (\xs ,t) = 
        \left(X_t(\xs), \frac{\partial}{\partial t}\Big|_{t}\right),
    \end{equation}
    where we use the splitting~\eqref{eq:TangentBundleOfMtimesB} of
    $T(\Sys \times \mathbb{R})$ and the canonical constant vector
    field on the ``bath'' $\Bad = \mathbb{R}$. For simplicity, we
    assume that $\cc{X}$ has a complete flow $\Psi_t$. Then the open
    time evolution for initial condition $\xb = 0$ of the bath is
    \begin{equation}
        \label{eq:PhiNulltTDVectorField}
        \Phi^0_t(\xs) = \pr_{\System} \left( \Psi_t(\xs , 0)\right).
    \end{equation}
    But this is precisely the time evolution of the time-dependent
    vector field $X_t$, i.e.\ we have
    \begin{equation}
        \label{eq:FlowOfXt}
        \frac{\D}{\D t} \Phi^0_t = X_t \circ \Phi^0_t,
    \end{equation}
    as an easy and well-known computation shows. Thus the ordinary
    time evolution of a time-dependent vector field can be viewed as a
    particular case of an open time evolution in the sense of
    Definition~\ref{definition:OpenTimeEvolutionPure}.
\end{example}

In view of the yet to be found quantization of open dynamical systems
we consider now the effect of an open time evolution on the functions
$C^\infty(\Sys)$ as these will play the role of the observables
later. The following statement is obvious:
\begin{proposition}
    \label{proposition:OpenTimeEvolutionPureIsHom}
    Let $\xb \in \Bad$. Then 
    $(\Phi^{\xb}_t)^*: C^\infty(\Sys) \longrightarrow C^\infty(\Sys)$
    is a $^*$-homomorphism for every $t \in \mathbb{R}$ and we have
    \begin{equation}
        \label{eq:PullBackWithOTE}
        (\Phi^{\xb}_t)^* 
        = (\id \tensor \delta_{\xb})
        \circ \Psi_t^* \circ \pr_{\System}^*.
    \end{equation}
\end{proposition}
Here $\delta_{\xb}: C^\infty(\Sys) \longrightarrow \mathbb{C}$ denotes
the $\delta$-functional at $\xb$, i.e.\ the evaluation of a function
at the point $\xb$. Moreover, $\id \tensor \delta_{\xb}$ is the
induced map
\begin{equation}
    \label{eq:idtensoromega}
    \id \tensor \delta_{\xb}:
    C^\infty(\Sys) \tensor C^\infty(\Bad) 
    = C^\infty(\Sys \times \Bad) \longrightarrow C^\infty(\Sys),
\end{equation}
where $\tensor$ denotes the completed projective tensor product.  Note
that the involved Fr\'{e}chet spaces are nuclear anyway.

Though Proposition~\ref{proposition:OpenTimeEvolutionPureIsHom} is a
trivial reformulation of the definition of $\Phi_t^{\xb}$ it gives a
new point of view: to this end, recall that a linear functional
$\omega_0: C^\infty(M) \longrightarrow \mathbb{C}$ is called
\emph{positive} if $\omega_0(\cc{f}f) \ge 0$ for all functions $f \in
C^\infty(M)$.  Similarly, we can define a positive functional on
matrix-valued functions $M_n(C^\infty(M))$. Having the notion of
positive linear functionals we can define positive algebra elements by
setting that $f \in C^\infty(M)$ is \emph{positive} if $\omega_0(f)
\ge 0$ for all positive linear functionals $\omega_0$. Then it is a
true but slightly non-trivial fact that $f$ is positive iff $f(p) \ge
0$ for all points $p \in M$. The same holds for matrix-valued
functions: a function $F \in M_n(C^\infty(M))$ is positive iff $F(p)$
is a positive semi-definite matrix for all $p \in M$. Note that in our
approach, this is \emph{not} a definition but a consequence of the
more algebraic definition.  Finally, a linear map $\phi: C^\infty(M)
\longrightarrow C^\infty(N)$ is called \emph{positive} if it maps
positive functions to positive functions. More important is the notion
of a completely positive map: $\phi$ is called \emph{completely
  positive} if all the canonical extensions $\phi: M_n(C^\infty(M))
\longrightarrow M_n(C^\infty(N))$ are positive maps for $n \in
\mathbb{N}$.  Clearly, this is the standard definition valid for every
$^*$-algebra over the complex numbers $\mathbb{C}$, see e.g.\
\cite{schmuedgen:1990a} for a detailed exposition and
\cite[App.~B]{bursztyn.waldmann:2001a} for a discussion of the case of
smooth functions.

Now we come back to our particular situation: while $\Phi_t^*$ and
$\pr_{\System}^*$ are canonically given $^*$-homomorphisms of the
$^*$-algebras of smooth functions and hence completely positive maps,
the map $\id \tensor \delta_{\xb}$ can also be interpreted as a
positive (and in fact completely positive) map which coincides with a
$^*$-homomorphism $\iota_{\xb}^*$ ``by accident''. In particular, we
can replace the positive functional $\delta_{\xb}$ by any, not
necessarily pure state $\omega_0$ of $C^\infty(\Bad)$, that is, a
positive linear normalized functional $\omega_0: C^\infty(\Bad)
\longrightarrow \C$. This yields the following, more general
definition of an open time evolution:
\begin{definition}[Open time evolution, general case]
    \label{definition:OpenTimeEvolutionGeneral}
    For any state $\omega_0: C^\infty(\Bad) \longrightarrow
    \mathbb{C}$ of the bath, the open time evolution of $\Sys$ with
    respect to the total time evolution $\Psi_t$ and the state
    $\omega_0$ is given by
    \begin{equation}
        \label{eq:OpenTimeEvolutionGeneral}
        (\Phi_t^{\omega_0})^*
        = (\id \tensor \omega_0) \circ \Psi_t^* \circ \pr_{\System}^*.
    \end{equation}
\end{definition}
\begin{remark}
    \label{remark:PosFunctBorelMeasure}
    Any positive functional $\omega_0: C^\infty(\Bad) \longrightarrow
    \mathbb{C}$ is actually a positive Borel measure with compact
    support, see e.g.\ \cite[App.~B]{bursztyn.waldmann:2001a}: for
    continuous functions this is the famous Riesz Representation
    Theorem, see e.g.\ \cite[Thm.~2.14]{rudin:1987a}, which can be
    shown to extend to the smooth setting.  Therefore, any state
    $\omega_0: C^\infty(\Bad) \longrightarrow \mathbb{C}$ is
    automatically continuous with respect to the smooth topology. Thus
    the map $\id \otimes \omega_0$ extends to the completed tensor
    product making the above expression in
    \eqref{eq:OpenTimeEvolutionGeneral} well-defined.
\end{remark}
The notation $(\Phi_t^{\omega_0})^*$ is of course only symbolic as
there is clearly \emph{no longer an underlying map of manifolds}. With
this definition we shifted the focus to the observable algebra rather
than the underlying geometry.
\begin{proposition}
    \label{proposition:PhiOmegaCP}
    For any state $\omega_0$ of the bath, the open time evolution
    $(\Phi_t^{\omega_0})^*: C^\infty(\Sys) \longrightarrow
    C^\infty(\Sys)$ is a completely positive map.
\end{proposition}
\begin{proof}
    Since $\Psi_t^*$ and $\pr_{\System}^*$ are $^*$-homomorphisms we
    only have to show that $\id \tensor \omega_0$ is a completely
    positive map from $C^\infty(\Sys \times \Bad)$ to
    $C^\infty(\Sys)$.  Thus let $F \in M_n(C^\infty(\Sys \times
    \Bad))$ be given and let $\xs \in \Sys$.  Then we have the
    embedding $\iota_{\xs}: \Bad \longrightarrow \Sys \times \Bad$
    whence
    \begin{equation}
        \label{eq:deltamomega}
        \delta_{\xs} \circ (\id \tensor \omega_0) 
        = \delta_{\xs} \tensor \omega_0 
        = \omega_0 \circ (\delta_{\xs} \tensor \id)
        = \omega_0 \circ \iota^*_{\xs}.
    \end{equation}
    Since $\iota_{\xs}^*$ is a $^*$-homomorphism, the composition
    $\omega_0 \circ \iota^*_{\xs}$ is still a positive functional and
    hence a completely positive map. Thus, applied to $F^*F$, we get a
    positive semi-definite matrix $\omega_0 \circ \iota^*_{\xs} (F^*F)
    = \delta_{\xs} \circ (\id \tensor \omega_0) (F^*F)$.  Since this
    is true for every point $\xs \in \Sys$, we have a positive element
    $(\id \tensor \omega_0)(F^*F) \in M_n(C^\infty(\Sys))$ proving the
    claim.
\end{proof}
\begin{remark}
    \label{remark:IntegrateBout}
    Since any positive functional $\omega_0: C^\infty(\Bad)
    \longrightarrow \mathbb{C}$ is actually a positive Borel measure
    with compact support, the map $\id \tensor \omega_0$ indeed means
    to integrate over the bath degrees of freedom with respect to a
    measure specified by $\omega_0$.
\end{remark}
\begin{remark}
    \label{remark:OTENotHomomorphism}
    Note also that in the case of a $\delta$-functional instead of an
    arbitrary state $\omega_0$, the open time evolution actually is a
    $^*$-homomorphism, in contrast to the case of arbitrary states.
    However, in general, $(\Phi_t^{\omega_0})^*$ is just a completely
    positive map without any further nice algebraic features.
\end{remark}

While up to now we have considered arbitrary dynamical systems, we
shall pass to more specific ones: we assume to have a Hamiltonian
dynamics on the total space of the system and the bath. In more
detail, we choose the rather general setting of Poisson geometry to
formulate Hamiltonian dynamics. This framework contains in particular
any symplectic phase space such as coadjoint orbits, cotangent bundles
or Kähler manifolds. However, also the dual of a Lie algebra is a
(linear) Poisson manifold which is important when dealing with
symmetries, see e.g.\ \cite[Chap.~3 \& Chap.~4]{waldmann:2007a} for an
introduction.

Thus, let the state space of the system
$(\Sys,\pi_{\System}^{\phantom{\sharp}})$ and the one of the bath
$(\Bad,\pi_{\Bath}^{\phantom{\sharp}})$ be in addition Poisson
manifolds with Poisson structures $\pi_{\System}^{\phantom{\sharp}}$
and $\pi_{\Bath}^{\phantom{\sharp}}$. On the total system $\Sys \times
\Bad$ we choose the product Poisson structure
\begin{equation}
    \label{eq:TotalPoissonStructure}
    \pi
    = \pr^\sharp_{\System} \pi_{\System}^{\phantom{\sharp}}
    + \pr^\sharp_{\Bath} \pi_{\Bath}^{\phantom{\sharp}}.
\end{equation}
This means that for functions $f_{\System}, g_{\System} \in
C^\infty(\Sys)$ and $f_{\Bath}, g_{\Bath} \in C^\infty(\Bad)$ the
factorizing functions $f = f_{\System} \otimes f_{\Bath}$ and $g =
g_{\System} \otimes g_{\Bath}$ have the Poisson bracket
\begin{equation}
    \label{eq:PoissonBracketOnFactorizingFunctions}
    \{f, g\}
    =
    \{f_{\System}, g_{\System}\}_{\System} \otimes f_{\Bath} g_{\Bath}
    +
    f_{\System} g_{\System} \otimes \{f_{\Bath}, g_{\Bath}\}_{\Bath}.
\end{equation}

The dynamics of the total system is given by the Hamiltonian vector
field $X_H \in \gaminf{T(\Sys\times \Bad)}$ with respect to the
\emph{total Hamiltonian} $H \in \cinf{\Sys\times \Bad}$. Recall that
the Hamiltonian vector field is defined by $X_H = \{\cdot, H\}$. In
typical situations, the total Hamiltonian contains three parts: we
have the Hamiltonian $H_{\System}^{\phantom{*}} \in C^\infty(\Sys)$ of
the system alone, the Hamiltonian $H_{\Bath}^{\phantom{*}} \in
C^\infty(\Bad)$ of the bath alone, and an interaction Hamiltonian
$H_{\Interaction}^{\phantom{*}} \in C^\infty(\Sys \times \Bad)$ such
that the total Hamiltonian is
\begin{equation}
    \label{eq:TotalHamiltonian}
    H = 
    \pr_{\System}^* H_{\System}^{\phantom{*}}
    + \pr_{\Bath}^* H_{\Bath}^{\phantom{*}}
    + H_{\Interaction}^{\phantom{*}}.
\end{equation}
Then the \emph{total Hamiltonian time evolution} is the flow $\Phi_t:
\Sys \times \Bad \longrightarrow \Sys \times \Bad$ of $X_H$ which
we assume to be complete for simplicity and analogously to
Definition~\ref{definition:OpenTimeEvolutionGeneral} the open
Hamiltonian time evolution with respect to a given state of the bath
is defined as follows:
\begin{definition}[Classical open Hamiltonian time evolution]
    \label{definition:ClassOpenTimeEvoMicroscopSys}
    The classical open Hamiltonian time evolution of the system $\Sys$
    with respect to a total Hamiltonian time evolution $\Phi_t$ of
    $\Sys \times \Bad$ and a given state $\omega_0$ of the bath is
    given as the open time evolution
    \begin{equation}
        \label{eq:OpenHamiltonianTimeEvolution}
        (\Phi_t^{\omega_0})^*:
        C^\infty(\Sys) \longrightarrow C^\infty(\Sys)
    \end{equation}
    according to Definition~\ref{definition:OpenTimeEvolutionGeneral}.
\end{definition}
\begin{remark}
    \label{remark:OnlyCP}
    Again, unless we have special circumstances, the open Hamiltonian
    time evolution is just a completely positive map without any
    further algebraic features. In particular, there is \emph{no
      reason} that $(\Phi_t^{\omega_0})^*$ should preserve Poisson
    brackets, even for $\omega_0 = \delta_{\xb}$ being a pure state.
\end{remark}

%
%

\section{Deformation Quantization of Open Hamiltonian Systems}
\label{sec:DeformationQuantizationOpenHamiltonianSystems}

In this section we will establish the deformation quantized version of
the open Hamiltonian time evolution. To this end, we recall that a
\emph{formal star product} on a Poisson manifold $(M, \pi)$ is an
associative $\mathbb{C}[[\hbar]]$-bilinear multiplication
\begin{equation}
    \label{eq:StarProduct}
    f \star g = \sum_{r=0}^\infty \hbar^r C_r(f, g)
\end{equation}
for $f, g \in C^\infty(M)[[\hbar]]$ such that $C_0 (f, g) = fg$ is the
undeformed commutative product, $C_1(f, g) - C_1(g, f) = \I \{f, g\}$
with the Poisson bracket $\{\cdot,\cdot\}$, $1 \star f = f = f \star
1$ for the constant function $1$, and all $C_r$ are bidifferential
operators \cite{bayen.et.al:1978a}, see also \cite{waldmann:2007a} for
a pedagogical introduction. The reason that we chose formal star
products where a priori no convergence in $\hbar$ is controlled, is
that for this situation we have the powerful existence and
classification theorems of deformation quantization at hand.
Physically, of course, one would like to have convergence or at least
some asymptotic statements. In many examples this is possible but we
shall not enter this rather technical issue here any further.

In the sequel, the case where the star product $\star$ is
\emph{Hermitian} will be important, i.e.\
\begin{equation}
    \label{eq:HermitianStarProduct}
    \cc{f \star g} = \cc{g} \star \cc{f}
\end{equation}
for all $f, g \in C^\infty(M)[[\hbar]]$ where $\cc{\hbar} = \hbar$ is
treated as a \emph{real} quantity. This $^*$-involution will be
necessary to have the honest interpretation of the algebra
$\rk{C^\infty(M)[[\hbar]], \star}$ as observable algebra of the
quantum system corresponding to the classical system.

Having the observable algebra, it is natural to define the states in
the same way as classically: we use positive linear functionals. Now
however, we have to specify first what a \emph{positive formal series}
should be. Here we can rely on the following definition. A non-zero
real formal power series $a = \sum_{r=r_0}^\infty \hbar^r a_r \in
\mathbb{R}[[\hbar]]$ is called \emph{positive} if its lowest non-zero
component is positive, $a_{r_0} > 0$. In this case we write
$a\geq0$. This is a good definition for many reasons: if we view
formal series as arising from asymptotic
expansions then this is what remains from a positive function. More
algebraically, $\mathbb{R}[[\hbar]]$ becomes an \emph{ordered ring} by
this definition, hence we can rely on the rich and well-developed
theory of $^*$-algebras over ordered rings, see e.g.\
\cite{bursztyn.waldmann:2005b, waldmann:2005b} for an overview and
\cite[Chap.~7]{waldmann:2007a} for an introduction and further
references.

For star product algebras we can proceed analogously to the classical case and define a
$\mathbb{C}[[\hbar]]$-linear functional $\omega: C^\infty(M)[[\hbar]]
\longrightarrow \mathbb{C}[[\hbar]]$ to be \emph{positive} if
\begin{equation}
    \label{eq:omegaPositive}
    \omega(\cc{f} \star f) \ge 0
\end{equation}
for all $f \in C^\infty(M)[[\hbar]]$. It can be shown that it suffices
to check \eqref{eq:omegaPositive} for $f \in C^\infty(M)$ without
higher orders of $\hbar$. Analogously, we define positive linear
functionals for matrix-valued functions $F \in
M_n(C^\infty(M)[[\hbar]])$ where the star product is extended to
matrices in the usual way. Having positive functionals we define $f
\in C^\infty(M)[[\hbar]]$ or $F \in M_n(C^\infty(M)[[\hbar]])$ to be a
\emph{positive algebra element} if
\begin{equation}
    \label{eq:omegafOmegaF}
    \omega(f) \ge 0
    \quad
    \textrm{and}
    \quad
    \Omega(F) \ge 0
\end{equation}
for all positive functionals $\omega$ and $\Omega$, respectively.
Finally, a $\mathbb{C}[[\hbar]]$-linear map $\phi:
C^\infty(M)[[\hbar]] \longrightarrow C^\infty(N)[[\hbar]]$ between two
star product algebras on possibly different underlying manifolds is
called \emph{positive} if $\phi$ maps positive elements to positive
elements. Equivalently, $\phi$ is called positive if $\omega \circ
\phi$ is a positive functional on $C^\infty(M)[[\hbar]]$ for all
positive functionals $\omega$ on $C^\infty(N)[[\hbar]]$. The map
$\phi$ is called \emph{completely positive} if this is also true for
arbitrary matrix-valued functions, i.e.\ if $\phi^{(n)}:
M_n(C^\infty(M)[[\hbar]]) \longrightarrow M_n(C^\infty(N)[[\hbar]])$
is positive for all $n \in \mathbb{N}$. Note that even though these
definitions are in complete analogy to the classical situation, it is
nevertheless crucial to have a good notion of positive formal power
series in $\mathbb{R}[[\hbar]]$.
\begin{remark}
    \label{remark:StarAlgebraOverOrderedRings}
    It is clear that the above concepts generalize immediately to
    $^*$-algebras $\mathsf{A}$ over a ring $\ring{C} = \ring{R}(\I)$
    where $\ring{R}$ is an ordered ring and $\I$ is a square root of
    $-1$. Even though many of the following considerations generalize
    to this algebraic framework as well, we shall focus on the more
    particular situation of star products.
\end{remark}
\begin{remark}
    \label{remark:CompletelyPositiveMaps}
    In the following, completely positive maps will play a crucial
    role. It is easy to see that positive functionals are in fact
    completely positive maps. Also $^*$-homomorphisms are completely
    positive. Moreover, the composition of completely positive maps as
    well as convex combinations of completely positive maps are again
    completely positive. Finally, less evident but nevertheless true
    is the fact that the algebraic tensor product of completely positive maps is
    again completely positive. In general, this last statement is
    wrong for positive maps.
\end{remark}

To describe the positive $\mathbb{C}[[\hbar]]$-linear functionals of
$\rk{C^\infty(M)[[\hbar]],\star}$ one first notes that $\omega$ is
necessarily of the form
\begin{equation}
    \label{eq:omegaFormalSeries}
    \omega = \sum_{r=0}^\infty \hbar^r \omega_r
    \quad
    \textrm{with linear maps}
    \quad
    \omega_r: C^\infty(M) \longrightarrow \mathbb{C}.
\end{equation}
Then the positivity $\omega(\cc{f} \star f) \ge 0$ in the sense of
formal power series immediately implies that $\omega_0(\cc{f}f) \ge 0$
classically, i.e.\ $\omega_0$ is a positive $\mathbb{C}$-linear
functional. This raises the question whether every classical state
$\omega_0$ can be ``quantized'' into a state $\omega$ with respect to
the star product. In other words, we ask whether every classical state
is the classical limit of some quantum state. Physically, this is
absolutely necessary as quantum theory is believed to be the more
fundamental description of nature. Fortunately, we can rely on the
following theorem \cite{bursztyn.waldmann:2005a}, even for the case of
matrices. But first we give a definition which shall simplify the
further considerations.

\begin{definition}[Square preserving map]
    \label{definition:SquarePreservingMap}
    A $\mathbb{C}[[\hbar]]$-linear map $\PresSquares{S} = \id +
    \sum_{r=1}^\infty\hbar^r\PresSquares{S}_r: C^\infty(M)[[\hbar]]$
    $\longrightarrow C^\infty(M)[[\hbar]]$ with differential operators
    $\PresSquares{S}_r$, $\PresSquares{S}(1)=1$, and
    $\PresSquares{S}(\overline{f}) = \overline{\PresSquares{S}(f)}$ is
    called \emph{preserving squares} with respect to $\star$, if there
    are formal series of differential operators $D_{r, I}:
    C^\infty(M)[[\hbar]] \longrightarrow C^\infty(M)[[\hbar]]$ for
    $r \in \mathbb{N}_0$ and $I$ running over a finite range (possibly
    depending on $r$) such that
    \begin{equation}
        \label{eq:SMapSquaresToSquares}
        \PresSquares{S}(\cc{f} \star g)
        =
        \sum_{r=0}^\infty \hbar^r \sum_{I}
        \overline{D_{r, I}(f)} D_{r, I}(g)
    \end{equation}
    for all $f, g \in C^\infty(M)[[\hbar]]$.
\end{definition}
\begin{remark}
    It is fairly simple to see that a map preserving squares according
    to Definition~\ref{definition:SquarePreservingMap} is in fact a
    completely positive map from the quantized algebra
    $(C^\infty(M)[[\hbar]], \star)$ to the classical algebra
    $(C^\infty(M)[[\hbar]], \cdot)$ with the \emph{undeformed}
    product.
\end{remark}
\begin{theorem}
    \label{theorem:ExistPosDef}
    Given a Hermitian star product $\star$, there exists a globally
    defined map $\PresSquares{S}$ preserving squares with respect to
    $\star$.
\end{theorem}
\begin{proof}
    By \cite{bursztyn.waldmann:2005a} we know that for a Hermitian
    star product $\star$ on an open subset $U \subseteq \mathbb{R}^n$
    there exists a map preserving squares with respect to $\star$,
    denoted by
    \[\tag{$\ast$}
    \PresSquares{S} (\cc{f} \star g)
    =
    \sum_{r=0}^\infty \hbar^r \sum_I 
    \overline{D_{r, I}(f)} D_{r, I}(g).
    \]
    For the Poisson manifold $M$ with star product $\star$ we choose a
    \emph{finite} atlas. Note that we can always find an atlas
    consisting of $\dim(M) + 1$ not necessarily connected
    charts. Denote the domains of the charts by $U_\alpha \subseteq
    M$.  Next we choose a corresponding subordinate \emph{finite}
    quadratic partition of unity $\chi_\alpha \in C^\infty(M)$, i.e.\
    $\supp \chi_\alpha \subseteq U_\alpha$ and $\sum_\alpha
    \cc{\chi}_\alpha \chi_\alpha = 1$. Now let
    $\PresSquares{S}_\alpha$ be the locally available maps preserving
    squares with respect to $\star|_{U_\alpha}$ with corresponding
    locally defined differential operators $D_{r, I, \alpha}$. Then we
    set
    \[
    \PresSquares{S}(f)
    =
    \sum\nolimits_\alpha \cc{\chi}_\alpha \chi_\alpha
    \PresSquares{S}_\alpha\left(f\big|_{U_\alpha}\right).
    \]
    Clearly, this gives a globally well-defined formal series of
    differential operators with $\cc{\PresSquares{S}(f)} =
    \PresSquares{S}(\cc{f})$ and $\PresSquares{S}(1) = 1$. Moreover,
    since the star product is bidifferential, we have $(f \star
    g)|_{U\alpha} = f|_{U_\alpha} \star g|_{U_\alpha}$ and hence we
    can apply ($\ast$) to obtain
    \[
    \PresSquares{S}(\cc{f} \star g)
    =
    \sum_{r=0}^\infty \hbar^r \sum_{I, \alpha}
    \cc{\chi_\alpha D_{r, I, \alpha} (f)}
    \chi_\alpha D_{r, I, \alpha} (g).
    \]
\end{proof}
\begin{remark}
    \label{remark:CstarVersion}
    Recently, a $C^*$-algebraic version of this theorem was obtained
    for particular strict deformation quantizations in
    \cite{kaschek.neumaier.waldmann:2009a}.
\end{remark}
The proof of Theorem~\ref{theorem:ExistPosDef} immediately leads to
the following consequence.
\begin{corollary}
    For every Hermitian star product $\star$ on a Poisson manifold
    there exists an equivalent star product $\star'$ with the property
    that every classically positive linear functional $\omega_0$ is
    also positive with respect to $\star'$.
\end{corollary}
\begin{proof}
    This is now easy, as we take a map $\PresSquares{S}$ preserving
    squares with respect to $\star$. Then the star product $f \star' g
    = \PresSquares{S} (\PresSquares{S}^{-1}(f) \star
    \PresSquares{S}^{-1}(g))$ is easily shown to do the job.
\end{proof}

\begin{remark}
    \label{remark:CorrectionsToPositiveFunctional}
    Rephrasing the result from \cite{bursztyn.waldmann:2005a} in terms
    of Theorem~\ref{theorem:ExistPosDef} says that every classical
    positive linear functional $\omega_0$ can be deformed into a
    positive linear functional with respect to a Hermitian star
    product. Indeed, $\omega_0 \circ \PresSquares{S}$ will be such a
    deformation, even universal for all $\omega_0$ once
    $\PresSquares{S}$ is specified.  In general, correction terms in
    higher orders of $\hbar$ are necessary to obtain
    positivity. Moreover, they are by far not unique and neither is
    the map $\PresSquares{S}$.  This is of course to be expected, both
    from a physical and mathematical point of view. Finally, note that
    each term $\omega_0 \circ \PresSquares{S}_r$ is continuous in the
    smooth topology, since the classical functional $\omega_0$ is
    continuous and the differential operators $\PresSquares{S}_r$ are
    as well.
\end{remark}

After this discussion of states we also need a notion of time
evolution for star product algebras. Here we can rely on the following
facts. For a given Hamiltonian $H \in C^\infty(M)[[\hbar]]$, where we
might even allow for some $\hbar$-dependent correction terms, we
consider the Heisenberg equation
\begin{equation}
    \label{eq:HeisenbergEquation}
    \frac{\D}{\D t} f(t) = \frac{\I}{\hbar} [H, f(t)]_\star
\end{equation}
for $f(t) \in C^\infty(M)[[\hbar]]$. Note that the right-hand side is
a well-defined formal power series since the commutator vanishes in
zeroth order. For simplicity, we assume that the Hamiltonian vector
field corresponding to the zeroth order $H_0$ of $H$ has a complete
flow $\Phi_t$. In this case, one can show that
\eqref{eq:HeisenbergEquation} has a solution for all times with the
following properties: There exists a formal series of time-dependent
differential operators $T_t = \id + \sum_{r=1}^\infty \hbar^r
T_t^{(r)}$ on $M$ such that
\begin{equation}
    \label{eq:QuantizedTimeEvolution}
    \sze{A}{t} = \Phi_t^* \circ T_t:
    C^\infty(M)[[\hbar]] \longrightarrow C^\infty(M)[[\hbar]]
\end{equation}
is a one-parameter group of automorphisms of $\star$ with $f(t) =
\sze{A}{t} f$ being the unique solution of
\eqref{eq:HeisenbergEquation} with initial condition $f(0) = f$.
Moreover, $\sze{A}{t}$ commutes with the commutator $[H,
\cdot]_\star$ and we have conservation of energy $\sze{A}{t} H = H$ as
usual. Finally, if $\star$ is a Hermitian star product and $H =
\cc{H}$ a real Hamiltonian then $\sze{A}{t}$ is even a
$^*$-automorphism for each $t$. For details on this quantized version
of the classical time evolution we refer to
\cite[Sect.~6.3.4]{waldmann:2007a} and references therein.

After this preparatory discussion we come back to our original
situation of a coupled total system $\Sys \times \Bad$. As we already
have a nice separation of the total Poisson structure into the Poisson
structure of the system and the one of the bath, we shall require the
same feature also for the quantization. Thus, we assume to have
Hermitian star products $\starSys$ on $\Sys$ and $\starBad$ on $\Bad$,
respectively. Then this immediately induces a Hermitian star product
$\star = \starSys \tensor \starBad$ on $\Sys \times \Bad$ in such a
way that
\begin{equation}
    \label{eq:psSprBHoms}
    \left(C^\infty(\Sys)[[\hbar]],\starSys\right)
    \stackrel{\pr_{\System}^*}{\longrightarrow}
    \left(C^\infty(\Sys \times \Bad)[[\hbar]], \star\right)
    \stackrel{\pr_{\Bad}^*}{\longleftarrow}
    \left(C^\infty(\Bad)[[\hbar]], \starBad\right)
\end{equation}
are both $^*$-homomorphisms of the involved star products. On
factorizing functions we have
\begin{equation}
    \label{eq:StarFactorizingFunctions}
    f \star g
    =
    (f_{\System} \starSys g_{\System})
    \otimes
    (f_{\Bath} \starBad g_{\Bath}),
\end{equation}
where $f = f_{\System} \otimes f_{\Bath}$ and $g = g_{\System} \otimes
g_{\Bath}$ for $f_{\System}, g_{\System} \in C^\infty(\Sys)[[\hbar]]$
and $f_{\Bath}, g_{\Bath} \in C^\infty(\Bad)[[\hbar]]$. Clearly,
\eqref{eq:PoissonBracketOnFactorizingFunctions} becomes the first
order limit of \eqref{eq:StarFactorizingFunctions} in the commutators.
\begin{remark}
    \label{remark:ObservablesAreObservables}
    It will be crucial for our approach that the algebraic structure
    of the observables is a priori given and will stay untouched. The
    physical interpretation is, that whatever the time evolution will
    be, the way how certain quantities, the observables, are measured
    is \emph{independent} of any sort of dynamics but a purely
    \emph{kinematical} property of the physical system. Thus our star
    products $\star$, $\starSys$, and $\starBad$ will be given once
    and for all and not changed by the open time evolution. Note that
    this is not the only possibility to deal with open systems: in
    \cite{dito.turrubiates:2006a} the star product itself was modified
    in order to describe a damped harmonic oscillator.
\end{remark}

It is now rather obvious what a good definition of a quantized open
Hamiltonian time evolution in deformation quantization should be:
\begin{definition}[Quantized open Hamiltonian time evolution]
    \label{definition:QuantizedOpenTimeEvolution}
    Let $H \in C^\infty(\Sys \times \Bad)[[\hbar]]$ be a Hamiltonian
    with complete time evolution $\sze{A}{t}$ and let $\omega:
    C^\infty(\Bad)[[\hbar]] \longrightarrow \mathbb{C}[[\hbar]]$ be a
    positive $\mathbb{C}[[\hbar]]$-linear functional. Then the
    \emph{quantized open Hamiltonian time evolution} of $\Sys$ with
    respect to $\omega$ is
    \begin{equation}
        \label{eq:QuantizedOpenTimeEvolution}
        \sze{A}{t}^\omega = (\id \tensor \omega) \circ
        \sze{A}{t} \circ \pr_{\System}^*:
        C^\infty(\Sys)[[\hbar]]
        \longrightarrow
        C^\infty(\Sys)[[\hbar]].
    \end{equation}
\end{definition}
\begin{remark}
    \label{remark:CompletedTensorProduct}
    The above completed tensor product is understood order by order in
    $\hbar$. Thus we have to require that $\omega = \sum_{r=0}^\infty
    \hbar^r \omega_r$ is \emph{continuous} in each order of $\hbar$,
    i.e.\ each $\omega_r$ is a continuous linear functional with
    respect to the smooth topology. In view of
    Theorem~\ref{theorem:ExistPosDef} and
    Remark~\ref{remark:CorrectionsToPositiveFunctional} this seems to
    be a very reasonable assumption.
\end{remark}
\begin{remark}
    \label{remark:Quantizability}
    Putting Theorem~\ref{theorem:ExistPosDef},
    Remark~\ref{remark:CorrectionsToPositiveFunctional}, the existence
    of Hermitian star products in \cite{kontsevich:2003a}, and the
    existence of the quantum time evolution of Equation
    \eqref{eq:QuantizedTimeEvolution} together it is easy to see that
    any classical open Hamiltonian time evolution can be quantized
    into a quantized open Hamiltonian time evolution. Conversely, the
    classical limit of any quantized open Hamiltonian time evolution
    is a classical open Hamiltonian time evolution for the classical
    limit of the Hamiltonian and with respect to the classical limit
    of the quantum state by construction as
    $\sze{A}{t}^\omega=\rk{\Phi_t^{\omega_0}}^*+\mathcal{O}(\hbar)$.
\end{remark}

In view of Definition~\ref{definition:QuantizedOpenTimeEvolution}
it is tempting to believe that the quantized open Hamiltonian time
evolution $\sze{A}{t}^\omega$ is \emph{completely positive}. Indeed,
if we would have used the algebraic tensor product in
\eqref{eq:QuantizedOpenTimeEvolution} instead of the completed one
$\tensor$ in every order of $\hbar$, then this would be a trivial
statement: the algebraic tensor product of the completely positive
maps $\id$ and $\omega$ is again completely positive, and so is the
composition with the completely positive $^*$-homomorphisms
$\sze{A}{t}$ and $\pr^*_\System$. However, the crucial point is that
the Fr\'echet topology of the smooth functions and the $\hbar$-adic
topology originating from the ring ordering are not very well
compatible. In fact, it is not clear whether the completed tensor
product is completely positive or not.  Note that this is rather
different from the $C^*$-algebraic case where the completed projective
tensor product of completely positive maps is always completely
positive.  From that point of view, the following principal result on
the quantized open Hamiltonian time evolution is non-trivial:
\begin{theorem}
    \label{theorem:SomeQuantizedOpenIsCP}
    Let $\omega$ be a positive $\mathbb{C}[[\hbar]]$-linear functional
    on $\rk{C^\infty(\Bad)[[\hbar]],\starBad}$ of the form
    \begin{align}
        \label{eq:ClassPosLinFunctDeformed}
        \omega = \omega_0 \circ \PresSquares{S}    
    \end{align}
    with $\PresSquares{S}$ preserving squares with respect to
    $\starBad$. Then any quantized open Hamiltonian time evolution
    with respect to $\omega$ is completely positive.
\end{theorem}
\begin{proof}
    As $\pr^*_\System$ and $\sze{A}{t}$ are $^*$-homomorphisms, the
    only thing left to show is that $\id\tensor\omega$ is completely
    positive. We extend $\PresSquares{S}$ to matrices as usual.  For
    $F \in M_n(C^\infty(\Sys \times \Bad)[[\hbar]])$ we have
    \[
    \iota_{\xb}^* \left(
        (\id \tensor \PresSquares{S})(F^* \star F)
    \right)
    =
    \sum_{r=0} \hbar^r \sum_I
    \iota_{\xb}^*\left(D_{r, I} (F)\right)^*
    \starSys
    \iota_{\xb}^*\left(D_{r, I} (F)\right),
    \]
    since the restriction to $\xb \in \Bad$ commutes with the
    \emph{pointwise} products in \eqref{eq:SMapSquaresToSquares}. Now
    let $\mu: C^\infty(\Sys)[[\hbar]] \longrightarrow
    \mathbb{C}[[\hbar]]$ be a positive $\mathbb{C}[[\hbar]]$-linear
    functional with respect to $\starSys$. Then for every $\xb$
    \[
    \iota_{\xb}^* \left(
        (\mu \tensor \PresSquares{S})(F^* \star F)
    \right)
    =
    \sum_{r=0} \hbar^r \sum_I
    \mu \left(
        \iota_{\xb}^* \left(D_{r, I} (F)\right)^*
        \starSys
        \iota_{\xb}^* \left(D_{r, I} (F)\right)
    \right)
    \in M_n(\mathbb{C})[[\hbar]]
    \]
    is \emph{positive}. So if $\omega_0$ is classically positive, we
    conclude that $\omega_0 \circ (\mu \tensor \PresSquares{S}) (F^*
    \star F) \ge 0$. But $\omega_0 \circ (\mu \tensor \PresSquares{S})
    = \mu \circ (\id \tensor (\omega_0 \circ \PresSquares{S})) = \mu
    \circ (\id \tensor \omega)$. Thus, $\mu \circ (\id \tensor
    \omega)(F^* \star F)$ is positive for all positive functionals
    $\mu$. This implies that $(\id \tensor \omega)(F^* \star F)$ is a
    positive algebra element for all matrices $F$ and hence $\id
    \tensor \omega$ is a completely positive map as claimed.
\end{proof}
\begin{remark}
    \label{remark:MoreQuantumStatesPossible}
    The assertion of Theorem~\ref{theorem:SomeQuantizedOpenIsCP} is
    actually true for more quantum states than the ones of type
    \eqref{eq:ClassPosLinFunctDeformed}: we will see examples later on
    in Proposition~\ref{proposition:KMSCPAndContinuous}. We also note
    that a possible failure of the complete positivity of
    $\sze{A}{t}^\omega$ should be seen as an artifact of the rather
    fine (and not too physical) $\hbar$-adic topology of formal power
    series in $\hbar$. One would expect reasonable behaviour as soon
    as one enters a convergent regime like strict deformation
    quantization.
\end{remark}
\begin{remark}
    \label{remark:NotAutomorphismNoNONOOOO}
    In general, the quantized open Hamiltonian time evolution
    $\sze{A}{t}^\omega$ is no $^*$-automorphism of
    $\rk{\cinf{\Sys}[[\hbar]],\star_{\System}}$.  Furthermore, a close
    look at Equation \eqref{eq:QuantizedOpenTimeEvolution} shows that
    usually $\sze{A}{s}^\omega \circ \sze{A}{t}^\omega \neq
    \sze{A}{s+t}^\omega$ as expected from a microscopic system.
\end{remark}
\begin{remark}
    \label{remark:SuperSpinAndStuff}
    Using the notions of super manifolds and star products on super
    symplectic manifolds according to \cite{bordemann:2000a,
      eckel:2000a} one can easily extend our formalism to this
    framework. This way, one can incorporate spin systems.
\end{remark}

%
%

\section{Linearly Coupled Harmonic Oscillators I: Generalities}
\label{section:ExampleLinearlyCoupledHarmOszisI}

As an example, consider the well-known linear coupling of two
one-dimensional harmonic oscillators.  We shall describe a
one-dimensional harmonic oscillator as a Hamiltonian system $(M, \pi,
H)$, given by $M = T^*\R_q\simeq\R^2_{q,p}$, with Hamiltonian $H(q,p)
= \frac{1}{2m}p^2 + \frac{m\nu^2}{2}q^2$, where $m,\nu\in\R^+$. The
Poisson bracket is then determined by
\begin{align*}
  \gk{q, p} = 1,
  \qquad
  \gk{q, q} = 0 = \gk{p, p}.
\end{align*}
Now let us take $\Sys = M = \Bad$.  The Hamiltonian system $(\Sys
\times \Bad, \pi, H)$ describing the linearly coupled identical
harmonic oscillators is then given by the smooth manifold $\Sys \times
\Bad \simeq \R^2_{\qs,\ps} \times \R^2_{\qb,\pb}$, with the
corresponding Poisson bracket as given by Equation
\eqref{eq:PoissonBracketOnFactorizingFunctions}.  In the following, we
shall use the same symbols $\qs, \ps, \qb, \pb$ for the coordinate
functions on $\Sys$, $\Bad$, and $\Sys \times \Bad$, respectively, in
order to simplify our notation. In the same spirit, we simply write $H
= H_{\System}^{\phantom{*}} + H_{\Bath}^{\phantom{*}} +
H_\Interaction^{\phantom{*}}$ for the total Hamiltonian without the
explicit use of $\pr_{\System}^*$ and $\pr_{\Bath}^*$. For the
linearly coupled harmonic oscillators the interaction term is given by
$H_\Interaction^{\phantom{*}} = \frac{\kappa}{2}(\qs-\qb)^2$, with
$\kappa\in\R^+$ being the coupling constant.

Using the new and still global coordinate functions
\begin{align*}
    \qen = \frac{1}{\sqrt{2}}(\qs+\qb),
    \quad
    \pen = \frac{1}{\sqrt{2}}(\ps+\pb),
    \quad
    \qzn = \frac{1}{\sqrt{2}}(\qs-\qb),
    \quad
    \pzn = \frac{1}{\sqrt{2}}(\ps-\pb),
\end{align*}
we can bring the total Hamiltonian to normal form and find the well-known
expression
\begin{align}
  \label{eq:NormalHamiltonian}
  H
  =
  \frac{1}{2m}\left(\pen^2 + \pzn^2\right)
  +
  \frac{m\nu^2}{2} \qen^2
  +
  \frac{m\nu^2_\kappa}{2} \qzn^2
  \quad
  \textrm{with}
  \quad
  \nu_\kappa^2 = \nu^2+\frac{2\kappa}{m}.
\end{align}
The classical time evolution $\Phi_t$ is known to be a linear map for
all $t$ which we can express in matrix form as
{\small
  \begin{align*}
      \Phi_t&=\frac{1}{2}
      \begin{pmatrix}
          \cos(\nu t)+\cos({\nu_\kappa}t)
          &\frac{\sin(\nu t)}{m\nu}
          + \frac{\sin({\nu_\kappa}t)}{m{\nu_\kappa}}
          &\cos(\nu t)-\cos({\nu_\kappa}t)
          &\frac{\sin(\nu t)}{m\nu}
          - \frac{\sin({\nu_\kappa}t)}{m{\nu_\kappa}}
          \\
          -m(\nu\sin(\nu t)+{\nu_\kappa}\sin({\nu_\kappa}t))
          &\cos(\nu t) + \cos({\nu_\kappa}t)
          &-m(\nu\sin(\nu t) - {\nu_\kappa}\sin({\nu_\kappa}t))
          &\cos(\nu t)-\cos({\nu_\kappa}t)
          \\
          \cos(\nu t)-\cos({\nu_\kappa}t)
          &\frac{\sin(\nu t)}{m\nu}
          - \frac{\sin({\nu_\kappa}t)}{m{\nu_\kappa}}
          &\cos(\nu t)+\cos({\nu_\kappa}t)
          &\frac{\sin(\nu t)}{m\nu} 
          + \frac{\sin({\nu_\kappa}t)}{m{\nu_\kappa}}
          \\
          -m(\nu\sin(\nu t)-{\nu_\kappa}\sin({\nu_\kappa}t))
          &\cos(\nu t) - \cos({\nu_\kappa}t)
          &-m(\nu\sin(\nu t) + {\nu_\kappa}\sin({\nu_\kappa}t))
          &\cos(\nu t)+\cos({\nu_\kappa}t)
      \end{pmatrix}
\end{align*}}
with respect to the global linear coordinates $\qs$, $\ps$, $\qb$,
$\pb$.  Thus, the open time evolution $\Phi_t^{\omega_0}$ of the open
subsystem with regard to the state $\omega_0$ of the bath takes the
form
\begin{align}
    \label{eq:OpenFlowR4}
    \begin{split}
        \rk{\Phi_t^{\omega_0}}^*
        \begin{pmatrix}
            \qs 
            \\
            \ps 
        \end{pmatrix}
        &=
        \frac{1}{2}
        \begin{pmatrix}
            \cos(\nu t)+\cos({\nu_\kappa}t)
            &\frac{\sin(\nu t)}{m\nu}
            + \frac{\sin({\nu_\kappa}t)}{m{\nu_\kappa}}
            \\
            -m(\nu\sin(\nu t)+{\nu_\kappa}\sin({\nu_\kappa}t))
            &\cos(\nu t)+\cos({\nu_\kappa}t)
        \end{pmatrix}
        \begin{pmatrix}
            \qs
            \\
            \ps
        \end{pmatrix}
        \\
        &\quad+
        \frac{1}{2}
        \begin{pmatrix}
            \omega_0(\qb)\left(\cos(\nu t) - \cos({\nu_\kappa}t)\right)
            +
            \omega_0(\pb)\left(
                \frac{\sin(\nu t)}{m\nu}
                -\frac{\sin({\nu_\kappa}t)}{m\nu_\kappa}
            \right)
            \\
            -\omega_0(\qb)m \left(
                \nu\sin(\nu t)
                - {\nu_\kappa}\sin({\nu_\kappa}t)
            \right)
            +
            \omega_0(\pb)\left(\cos(\nu t) - \cos({\nu_\kappa}t)\right)
        \end{pmatrix}.
    \end{split}
\end{align}

Analogously to the classical case we shall use the normal coordinates
in order to simplify the computation of the quantum
time evolution of the total system. Moreover, it will be advantageous
to combine the real $\qen$, $\pen$, $\qzn$, and $\pzn$ into complex
coordinates which will play the role of annihilation and creation
``operators'' later on. We set
\begin{align}
  \label{eq:NormalToComplex}
  \begin{array}{l}
      \z{k} = \sqrt{m\nu_k}\qin{k}+\I\frac{1}{\sqrt{m\nu_k}}\pin{k},
      \\
      \zb{k} = \sqrt{m\nu_k}\qin{k}-\I\frac{1}{\sqrt{m\nu_k}}\pin{k}
  \end{array}
  \quad
  \textrm{and hence}
  \quad
  \begin{array}{l}
      \qin{k} = \frac{1}{2\sqrt{m\nu_k}}(\z{k}+\zb{k}),
      \\
      \pin{k} = \frac{\sqrt{m\nu_k}}{2\I}(\z{k}-\zb{k})
  \end{array}
\end{align}
for $k = 1, 2$ and $\nu_1 = \nu$, $\nu_2 = \nu_\kappa$. With respect
to these global coordinate functions on $M$ the total Hamiltonian can be
written as $H = \frac{\nu}{2}\z{1}\zb{1} + \frac{{\nu_\kappa}}{2}
\z{2}\zb{2}$.  For the Poisson brackets one obtains $\gk{\z{k}, \z{l}}
= 0 = \gk{\zb{k},\zb{l}}$ and $\gk{\z{k}, \zb{l}} = \frac{2}{\I}
\delta_{kl}$ for all $k, l = 1, 2$. The Hamiltonian for the system will
take a slightly more complicated form, namely
\begin{equation}
    \label{eq:HSystemInzEinszZweiVariablen}
    H_{\System}
    =
    \frac{\nu}{4} \cc{z}_1 z_1
    +
    \frac{\nu^2}{16 \nu_\kappa} (z_2 + \cc{z}_2)^2
    -
    \frac{\nu_\kappa}{16} (z_2 - \cc{z}_2)^2
    +
    \frac{\nu^2}{8\sqrt{\nu\nu_\kappa}}
    (z_1 + \cc{z}_1)(z_2 + \cc{z}_2)
    -
    \frac{\sqrt{\nu\nu_\kappa}}{8}
    (z_1 - \cc{z}_1)(z_2 - \cc{z}_2).
\end{equation}

On the other hand, we will also need ``factorizing'' complex
coordinates with respect to the original Darboux coordinates on the
system $\Sys$ and the bath $\Bad$. Hence we set
\begin{align}
  \label{eq:DarbouxToComplex}
  \begin{array}{l}
  \z{\System} = \sqrt{m\nu}\qs + \I\frac{1}{\sqrt{m\nu}}\ps,
  \\
  \zb{\System} = \sqrt{m\nu}\qs - \I\frac{1}{\sqrt{m\nu}}\ps
  \end{array}
  \quad
  \textrm{and}
  \quad
  \begin{array}{l}
  \z{\Bath} = \sqrt{m\nu}\qb + \I\frac{1}{\sqrt{m\nu}}\pb,
  \\
  \zb{\Bath} = \sqrt{m\nu}\qb - \I\frac{1}{\sqrt{m\nu}}\pb.
  \end{array}
\end{align}
In these coordinates, the Hamiltonians of the system and the bath are
given by $H_\System = \frac{\nu}{2}\z{\System}\zb{\System}$ and
$H_\Bath = \frac{\nu}{2}\z{\Bath}\zb{\Bath}$. The interaction term now
reads $H_{\Interaction} = \frac{\kappa}{8m\nu} \rk{\z{\System} +
  \zb{\System} - \z{\Bath} - \zb{\Bath}}^2$.  Again, for the Poisson
brackets one finds $\gk{\z{k}, \z{l}} = 0 = \gk{\zb{k},\zb{l}}$ and
$\gk{\z{k}, \zb{l}} = \frac{2}{\I} \delta_{kl}$ for all $k,l =
\System, \Bath$.

After these preparations, we can specify the star product on the total
algebra of observables. We take the Weyl-Moyal star product on the
total system $\Sys \times \Bad$ defined by
\begin{align}
    \label{eq:WeylFlach}
    \begin{split}
        f&\star^\weyl g
        \\
        &=
        \sum_{r=0}^\infty
        \sum_{l=0}^r
        \rk{-\frac{\I}{2}}^{r-l}
        \rk{\frac{\I}{2}}^l
        \frac{\hbar^r}{l!(r-l)!}
        \sum_{i_1, \ldots, i_r = 1}^2
        \frac{\partial^r f}
        {
          \partial q_{i_1} \cdots \partial q_{i_l}
          \partial p_{i_{l+1}}\cdots\partial p_{i_{r}}
        }
        \frac{\partial^{r}g}
        {          
          \partial p_{i_1}\cdots\partial p_{i_l}
          \partial q_{i_{l+1}}\cdots\partial q_{i_{r}}
        },
    \end{split}
\end{align}
for $f, g \in \cinf{\Sys \times \Bad}[[\hbar]]$, see e.g.\
\cite{bayen.et.al:1978a} and \cite[Chap.~5]{waldmann:2007a}.
\begin{remark}
    \label{remark:WeylUnique}
    The Weyl-Moyal star product on a flat symplectic phase space
    $\R^{2n}$ is uniquely determined by the requirement of invariance
    under the affine symplectic group. Under the usual quantization
    map into differential operators it corresponds to the total
    symmetrization, see e.g.\ \cite[Chap.~5]{waldmann:2007a} for a
    detailed discussion. We also note that $\star^\weyl =
    \star^\weyl_\System \tensor \star^\weyl_\Bath$ as required by our
    general framework.
\end{remark}

While the Weyl-Moyal star product is the most natural one with respect
to phase space symmetries, it has certain technical disadvantages:
when dealing with harmonic oscillators, for technical reasons it will
be more convenient to employ a Wick star product. Such a Wick star
product is no longer unique, but depends on the choice of a compatible
linear complex structure on the phase space which is nothing but the
choice of a harmonic oscillator.  Therefore, we will have different
Wick star products adapted to the various harmonic oscillators on
hand: either with or without the coupling. In detail, one passes from
the Weyl-Moyal star product to the Wick star product by means of an
equivalence transformation explicitly given by
\begin{equation}
    \label{eq:WickEquivalence}
    S = \exp\rk{\hbar \Delta}
    \quad
    \textrm{with}
    \quad
    \Delta = \sum_{k=1}^2
    \frac{\partial^2}{\partial z_k\partial \overline{z}_k}.    
\end{equation}
Then the Wick star product $\star^\wick$ is defined by $f \star^\wick
g = S ( S^{-1}(f) \star^\weyl S^{-1} (g))$. Explicitly, $\star^\wick$
is given by
\begin{equation}
    \label{eq:WickStarProduct}
    f \star^\wick g
    =
    \sum_{r=0}^\infty \frac{(2\hbar)^r}{r!}
    \sum_{i_1, \ldots, i_r = 1}^2
    \frac{\partial^r f}
    {\partial z_{i_1} \cdots \partial z_{i_r}}
    \frac{\partial^r g}
    {\partial \cc{z}_{i_1} \cdots \partial \cc{z}_{i_r}}
\end{equation}
for $f, g \in C^\infty(\mathbb{R}^4)[[\hbar]]$. Alternatively, we
ignore the coupling term and use the complex coordinates
$\z{\System}$, $\zb{\System}$ for the system and $\z{\Bath}$,
$\zb{\Bath}$ for the bath. This gives the two equivalence
transformations
\begin{equation}
    \label{eq:WickForSystemAndBathTrafo}
    S_{\System}
    = \exp\left(
        \hbar \frac{\partial^2}
        {\partial \z{\System} \partial \zb{\System}}
    \right)
    \quad
    \textrm{and}
    \quad
    S_{\Bath}
    = \exp\left(
        \hbar \frac{\partial^2}
        {\partial \z{\Bath} \partial \zb{\Bath}}
    \right),
\end{equation}
acting on functions on $\Sys$ and $\Bad$, respectively. Analogously to
\eqref{eq:WickStarProduct} we get Wick star products
$\star^\wick_{\System}$ and $\star^\wick_{\Bath}$ for the system and
the bath, respectively. Since we ignored the coupling terms in the
definition of the latter two Wick star products, we have
\begin{equation}
    \label{eq:WickIsNotWickWick}
    S \neq S_{\System} \tensor S_{\Bath}
    \quad
    \textrm{and}
    \quad
    \star^\wick
    \neq
    \star^\wick_\System\tensor\star^\wick_\Bath.
\end{equation}

The total time evolution with respect to $\star$ and $H$ can actually be
calculated in a much easier way than by solving the corresponding
evolution equation \eqref{eq:HeisenbergEquation}: we first compute the
time evolution with respect to the Wick star product $\star^\wick$,
which turns out to be simple, and then transform the time evolved
observables back using $S$.

The total time evolution $\sze{A}{t}^\wick$ with respect to the Wick
star product is determined by
\begin{equation}
    \label{eq:WickTimeEvolution}
    \frac{\D}{\D t}\sze{A}{t}^\wick f
    =
    \frac{\I}{\hbar}\ek{H, \sze{A}{t}^\wick f}_{\star^\wick}
    =
    \gk{\sze{A}{t}^\wick f,H}
\end{equation}
for $f \in C^\infty(\Sys \times \Bad)[[\hbar]]$ due to the fact that
$H = \frac{\nu}{2}\zb{1}\z{1} + \frac{{\nu_\kappa}}{2}
\zb{2}\z{2}$. It immediately follows that the time evolution is just
the \emph{classical} one, i.e.\  $\sze{A}{t}^\wick = \Phi_t^*$, and no
higher order correction terms arise. But then it is clear that the
time evolution with respect to $\star$ is given by conjugation with
$S$ since $SH=H+c$ with a constant
$c=\hbar\frac{\nu+\nu_\kappa}{2}$. Hence, we have
\begin{align}
    \label{eq:DefTimeEvolTotal}
    \sze{A}{t} = S^{-1} \circ \Phi_t^* \circ S.
\end{align}
As a consequence we immediately obtain the following result for the
open time evolution with respect to the Weyl-Moyal star product:
\begin{proposition}
    \label{proposition:OpenTimeEvolutionWeyl}
    The deformed time evolution of the open subsystem with respect to
    the functional $\omega$ is given by
    \begin{equation}
        \label{eq:DeformedOpenTimeEvolution}
        \sze{A}{t}^\omega
        =
        (\id\tensor\omega) 
        \circ S^{-1} \circ \Phi_t^* \circ S
        \circ \pr_{\System}^*.
    \end{equation}
\end{proposition}
\begin{remark}
    The $^*$-automorphism $\sze{A}{t}$ obviously restricts to the
    polynomials $\Pol\rk{\Sys \times \Bad}[\hbar]$. Thus, being only
    interested in polynomial observables leads to a convergent
    formulation of the deformed time evolution of the open harmonic
    oscillator if the quantized state $\omega$ used to reduce the
    total dynamics gives a finite order in $\hbar$ for every
    polynomial on the bath. This will be the case for the deformed
    $\delta$-functionals as well as for the KMS functionals in
    Section~\ref{section:ExampleLinearlyCoupledHarmOszisII}. Thence,
    here we recover the usual quantum mechanical formulation including
    the convergence in $\hbar$.
\end{remark}

To further illustrate the above situation we compute the open
time evolution of some specific observables of the system. Here we
still allow for a general state $\omega$.

As a first step we calculate the total quantum time evolutions of the
total system for $\qs$ and $\ps$. To do so, we will have to evaluate
the chain of maps \eqref{eq:DeformedOpenTimeEvolution} applied to
these observables.  First we note that
\begin{equation}
    \label{eq:SonqAndp}
    S \qs = \qs = S^{-1} \qs
    \quad
    \textrm{and}
    \quad
    S \ps = \ps = S^{-1} \ps.
\end{equation}
Then the classical time evolution is \emph{linear} whence applying the
transformation $S^{-1}$ again does not give additional terms. We
conclude that
\begin{equation}
    \label{eq:QuantumTimeEvolvqsps}
    \sze{A}{t} \qs = \Phi_t^* \qs
    \quad
    \textrm{and}
    \quad
    \sze{A}{t} \ps = \Phi_t^* \ps.
\end{equation}
For the Hamiltonian $H_{\System}$ of the system the calculation is slightly more
complicated: First we note that applying $S$ yields an additional
constant, namely
\begin{equation}
    \label{eq:SHSystem}
    S H_{\System}
    =
    H_{\System}
    + \frac{\hbar}{4}
    \left(\nu + \nu_\kappa - \frac{\kappa}{m\nu_\kappa}\right).
\end{equation}
Now the total classical time
evolution of $H_{\System}$ is quite complicated and can be computed
most easily from $\Phi_t^* z_1 = \exp(-\I t \nu) z_1$ and $\Phi_t^*
z_2 = \exp(- \I t \nu_\kappa) z_2$ and
\eqref{eq:HSystemInzEinszZweiVariablen}. The remarkable fact is now
that
\begin{equation}
    \label{eq:DeltaHSystem}
    \Delta \Phi_t^* H_{\System}
    = 
    \Delta H_{\System}
    =
    \frac{1}{4}
    \left(\nu + \nu_\kappa - \frac{\kappa}{m\nu_\kappa}\right)
\end{equation}
for all $t$. Thus applying $S^{-1}$ to $\Phi_t^* H_{\System}$ gives
$\Phi_t^* H_{\System}$ minus the same constant as we obtained in
\eqref{eq:SHSystem}. We conclude that also for the Weyl star product
\begin{equation}
    \label{eq:TotalQuantumHSystem}
    \sze{A}{t} H_{\System} = \Phi_t^* H_{\System}.
\end{equation}
Replacing the complex coordinates and their (simple) time evolution
by the original real coordinates we get the explicit total
classical and hence also quantum time evolutions for $\qs$, $\ps$, and
$H_{\System}^{\phantom{*}}$
\begin{align}
    \sze{A}{t} \qs = \Phi_t^*\qs
    &=
    \frac{1}{2}
    \left(
        \cos(\nu t) + \cos({\nu_\kappa}t)
    \right) \qs
    +
    \left(
        \frac{\sin(\nu t)}{2m\nu}
        +
        \frac{\sin({\nu_\kappa}t)}{2m{\nu_\kappa}}
    \right) \ps
    \nonumber \\
    &\quad+
    \frac{1}{2}
    \left(
        \cos(\nu t) - \cos({\nu_\kappa}t)
    \right) \qb
    +
    \left(
        \frac{\sin(\nu t)}{2 m\nu}
        -
        \frac{\sin({\nu_\kappa}t)}{2 m{\nu_\kappa}}
    \right) \pb,
    \label{eq:TotalTimeEvolutionOfqs}
\end{align}
\begin{align}
    \sze{A}{t} \ps = \Phi_t^* \ps
    &=
    - \frac{m}{2}
    \left(
        \nu\sin(\nu t) + {\nu_\kappa}\sin({\nu_\kappa}t)
    \right) \qs
    +
    \frac{1}{2}
    \left(
        \cos(\nu t) + \cos({\nu_\kappa}t)
    \right) \ps
    \nonumber \\
    &\quad-
    \frac{m}{2}
    \left(\nu\cos(\nu t) - {\nu_\kappa}\cos({\nu_\kappa}t)\right) \qb
    +
    \frac{1}{2} \left(\cos(\nu t) - \cos({\nu_\kappa}t)\right) \pb,
    \label{eq:TotalTimeEvolutionOfps}
\end{align}
and
\begin{align}
    \sze{A}{t} H_{\System}^{\phantom{*}} = \Phi_t^*H_{\System}^{\phantom{*}}
    &=
    \rk{\frac{m}{8}(\nu\sin(\nu t)+\nu_\kappa\sin({\nu_\kappa}t))^2
      +\frac{m\nu^2}{8}(\cos(\nu t)+\cos({\nu_\kappa}t))^2}q_{\System}^2
    \nonumber \\
    &\quad+
    \rk{\frac{1}{8m}(\cos(\nu t)+\cos({\nu_\kappa}t))^2
      +\frac{m\nu^2}{8}\rk{\frac{\sin(\nu t)}{m\nu}
        +\frac{\sin({\nu_\kappa}t)}{m{\nu_\kappa}}}^2}p_{\System}^2
    \nonumber \\
    &\quad+
    \rk{\frac{m}{8}(\nu\sin(\nu
      t)-{\nu_\kappa}\sin({\nu_\kappa}t))^2 +\frac{m\nu^2}{8}(\cos(\nu
      t)-\cos({\nu_\kappa}t))^2}q_{\Bath}^2
    \nonumber \\
    &\quad+
    \rk{\frac{1}{8m}(\cos(\nu t)-\cos({\nu_\kappa}t))^2
      +\frac{m\nu^2}{8}\bigg(\frac{\sin(\nu t)}{m\nu}
      -\frac{\sin({\nu_\kappa}t)}{m{\nu_\kappa}}\bigg)^2}p_{\Bath}^2
    \nonumber \\
    &\quad+
    \bigg(-\frac{1}{4}(\nu\sin(\nu t)
    +{\nu_\kappa}\sin({\nu_\kappa}t))
    (\cos(\nu t) +\cos({\nu_\kappa}t))
    \nonumber \\
    &\qquad+
    \frac{m\nu^2}{4}\bigg(\frac{\sin(\nu t)}{m\nu}
    +\frac{\sin({\nu_\kappa}t)}{m{\nu_\kappa}}\bigg)
    (\cos(\nu t)+\cos({\nu_\kappa}t))\bigg)\qs\ps
    \nonumber \\
    &\quad+
    \bigg(\frac{m}{4}(\nu\sin(\nu t)
    +{\nu_\kappa}\sin({\nu_\kappa}t))(\nu\sin(\nu t)
    -{\nu_\kappa}\sin({\nu_\kappa}t))
    \nonumber \\
    &\qquad+
    \frac{m\nu^2}{4}(\cos(\nu t)+\cos({\nu_\kappa}t))
    (\cos(\nu t)-\cos({\nu_\kappa}t))\bigg)\qs\qb
    \nonumber \\
    &\quad+
    \bigg(-\frac{1}{4}(\cos(\nu t)-\cos({\nu_\kappa}t))
    \bigg(\frac{\sin(\nu t)}{m\nu}
    +\frac{\sin({\nu_\kappa}t)}{m{\nu_\kappa}}\bigg)
    \nonumber \\
    &\qquad+
    \frac{m\nu^2}{4}(\cos(\nu t)+\cos({\nu_\kappa}t))
    \bigg(\frac{\sin(\nu t)}{m\nu}
    -\frac{\sin({\nu_\kappa}t)}{m{\nu_\kappa}}\bigg)\bigg)\qs\pb
    \nonumber \\
    &\quad+
    \bigg(-\frac{1}{4}(\cos(\nu t)+\cos({\nu_\kappa}t))
    (\nu\sin(\nu t) -{\nu_\kappa}\sin({\nu_\kappa}t))
    \nonumber \\
    &\qquad+
    \frac{m\nu^2}{4}(\cos(\nu t)-\cos({\nu_\kappa}t))
    \bigg(\frac{\sin(\nu t)}{m\nu}
    +\frac{\sin({\nu_\kappa}t)}{m{\nu_\kappa}}\bigg)\bigg)\ps\qb
    \nonumber \\
    &\quad+
    \bigg(\frac{1}{4m}(\cos(\nu t)+\cos({\nu_\kappa}t))
    (\cos(\nu t)-\cos({\nu_\kappa}t))
    \nonumber \\
    &\qquad+
    \frac{m\nu^2}{4} \bigg(\frac{\sin(\nu t)}{m\nu}
    +\frac{\sin({\nu_\kappa}t)}{m{\nu_\kappa}}\bigg)
    \bigg(\frac{\sin(\nu t)}{m\nu}
    -\frac{\sin({\nu_\kappa}t)}{m{\nu_\kappa}}\bigg)\bigg)\ps\pb
    \nonumber \\
    &\quad+
    \bigg(-\frac{1}{4}(\nu\sin(\nu t)
    -{\nu_\kappa}\sin({\nu_\kappa}t))
    (\cos(\nu t)-\cos({\nu_\kappa}t))
    \nonumber \\
    &\qquad+
    \frac{m\nu^2}{4}(\cos(\nu t)-\cos({\nu_\kappa}t))
    \bigg(\frac{\sin(\nu t)}{m\nu}
    -\frac{\sin({\nu_\kappa}t)}{m{\nu_\kappa}}\bigg)\bigg)\qb\pb.
    \label{eq:TotalTimeEvolutionOfHSystem}
\end{align}
The reason for transforming the time evolved observables back to the
Darboux coordinate functions $\qs$, $\ps$, $\qb$, and $\pb$ is not
just an addiction to extensive exercise: It is in these
variables where we can apply the final map $\id \tensor \omega$ needed
for the open time evolution, where $\omega$ is a state of the bath
with respect to $\star^\weyl_{\Bath}$. The procedure is very simple:
we will have to replace all bath variables by their expectation values
with respect to $\omega$, i.e.\ $\qb$ is to be replaced by
$\omega(\qb)$, $\qb\pb$ is replaced by $\omega(\pb\pb)$ et cetera. We will
not write down the explicit formulas as these are now obtained from
\eqref{eq:TotalTimeEvolutionOfqs}, \eqref{eq:TotalTimeEvolutionOfps},
and \eqref{eq:TotalTimeEvolutionOfHSystem} just by copying.
\begin{remark}
    \label{remark:QuantumVersusClassical}
    Note that for these observables, the open time evolutions in the
    classical and quantum regime only differ by the (possibly)
    different expectation values with respect to $\omega$ and its
    classical limit $\omega_0$. In general, we have to expect
    additional quantum corrections from the total time evolution as
    well.
\end{remark}

%
%

\section{Linearly Coupled Harmonic Oscillators II: Examples}
\label{section:ExampleLinearlyCoupledHarmOszisII}

The first example of a state for the bath is a deformation of the
$\delta$-functional. Thus, fix a point $(\qb_0, \pb_0)$ in the bath and
consider $\delta_{(\qb_0, \pb_0)}$. For the Weyl-Moyal star product
this will no longer be a positive functional, see
e.g.\ \cite[Sect.~7.1.3]{waldmann:2007a}. However, for the Wick star
product $\star^\wick_\Bath$ on the bath the $\delta$-functional will
be positive without corrections. Thus using the equivalence
transformation $S_\Bath$ we obtain a positive functional
$\delta_{(\qb_0, \pb_0)} \circ S_\Bath$ with respect to the Weyl-Moyal
star product. Note that the equivalence transformation $S_\Bath$ is
precisely a map preserving squares with respect to the Weyl-Moyal star
product which is evident from the explicit formula for
$\star^\wick_\Bath$. In fact, this was the first example of a map
preserving squares which is also heavily used in the proofs in
\cite{bursztyn.waldmann:2005a}. More physically speaking,
$\delta_{(\qb_0, \pb_0)} \circ S_\Bath$ corresponds to a coherent state
localized around the point $(\qb_0, \pb_0)$.

For this particular state we note that for the observables at most
linear in $\qb$ and $\pb$ the operator $S_\Bath$ does not have a
non-trivial effect. Moreover, for the quadratic terms $q_{\Bath}^2$,
$p_{\Bath}^2$, and $\qb\pb$ the operator $S_\Bath$ only gives a
correction term in first order of $\hbar$. Explicitly, we obtain
\begin{equation}
    \label{eq:SbOnStuff}
    S_\Bath^{\phantom{*}} \qb = \qb,
    \quad
    S_\Bath^{\phantom{*}} \pb = \pb,
    \quad
    S_\Bath^{\phantom{*}}(\qb\pb) = \qb\pb,
\end{equation}
\begin{equation}
    \label{eq:SbOnOtherStuff}
    S_\Bath^{\phantom{*}} q_{\Bath}^2
    = q_{\Bath}^2 + \hbar\frac{1}{2m\nu},
    \quad
    \textrm{and}
    \quad
    S_\Bath^{\phantom{*}} p_{\Bath}^2
    = p_{\Bath}^2 + \hbar\frac{m\nu}{2}.
\end{equation}

From these computations we see that the open time evolutions with
respect to $\delta_{(\qb_0, \pb_0)} \circ S_\Bath$ are given by
\begin{equation}
    \label{eq:DeltaEvolutionOnqs}
    \sze{A}{t}^{\delta_{\qb_0,\pb_0} \circ S_{\Bath}}(\qs)
    =
    \left(\Phi_t^{\delta_{\qb_0,\pb_0}}\right)^* \qs,
\end{equation}
\begin{equation}
    \label{eq:DeltaEvolutionOnps}
    \sze{A}{t}^{\delta_{\qb_0,\pb_0} \circ S_\Bath}(\ps)
    =
    \left(\Phi_t^{\delta_{\qb_0,\pb_0}}\right)^* \ps,
\end{equation}
\begin{align}
    \label{eq:DeltaEvolutionOnHSystem}
    &\sze{A}{t}^{\delta_{\qb_0,\pb_0} \circ S_\Bath} H_{\System}
    =\left(\Phi_t^{\delta_{\qb_0,\pb_0}}\right)^* H_{\System}\nonumber \\
    &\quad
    +\frac{\hbar}{16}\left(
        \frac{1}{\nu} \left(
            \nu\sin(\nu t) - {\nu_\kappa} \sin({\nu_\kappa}t)
        \right)^2
        + \nu \left(
            \sin(\nu t) - \frac{\nu}{{\nu_\kappa}}\sin({\nu_\kappa}t)
        \right)^2
        +2\nu\rk{\cos(\nu t)-\cos(\nu_\kappa t)}
    \right).
\end{align}
\begin{remark}
    The classical open time evolutions of $\qs$, $\ps$, and
    $H_{\System}^{\phantom{*}}$ in \eqref{eq:DeltaEvolutionOnqs},
    \eqref{eq:DeltaEvolutionOnps}, and
    \eqref{eq:DeltaEvolutionOnHSystem} are obtained by replacing the
    functions $\qb$, $\pb$, and their powers in the Equations
    \eqref{eq:TotalTimeEvolutionOfqs},
    \eqref{eq:TotalTimeEvolutionOfps}, and
    \eqref{eq:TotalTimeEvolutionOfHSystem} by their values at $\qb_0$ and
    $\pb_0$.
\end{remark}
\begin{remark}
    \label{remark:NonclassicalTerms}
    The deformation of the $\delta$-functional necessary in order to
    ensure complete positivity leads to non-classical components of
    the open time evolution.
\end{remark}

Next we will study quantized states fulfilling a formal KMS condition,
corresponding to ``thermal equilibrium states'' of the bath.

To this end, we first recall that for every symplectic star product
$\star$ for $C^\infty(M)[[\hbar]]$ there is a unique trace functional
\begin{equation}
    \label{eq:TraceFunctional}
    \tr: C^\infty_0(M)[[\hbar]] \longrightarrow \mathbb{C}[[\hbar]],
\end{equation}
i.e.\ $\tr(f \star g) = \tr(g \star f)$. Choosing the normalization of
$\tr$ appropriately one obtains a \emph{positive} trace, see
e.g.\ \cite[Sect.~6.3.5]{waldmann:2007a} for a detailed discussion and
references. For the Weyl-Moyal star product, the trace is known to be
\begin{equation}
    \label{eq:TraceForWeylMoyal}
    \tr(f) = \int_{\mathbb{R}^{2n}} f(x) \D^{2n}x,
\end{equation}
i.e.\ the integration with respect to the Liouville volume. In fact, it
can be shown that in the symplectic case the lowest order of $\tr$ is
necessarily of this form: it is just the integration over the whole
manifold with respect to the Liouville volume.

The second ingredient we need is the $\star$-exponential $\Exp$, as
introduced in \cite{bayen.et.al:1978a}. Instead of defining the
exponential function by means of the series, the following approach
favoured in \cite{bordemann.roemer.waldmann:1998a}, see also
\cite[Sect.~6.3.1]{waldmann:2007a}, will be used. For $H \in
C^\infty(M)[[\hbar]]$ one defines $\Exp(\beta H) \in
C^\infty(M)[[\hbar]]$ to be the unique solution of the differential
equation
\begin{equation}
    \label{eq:ExpODE}
    \frac{\D}{\D \beta} \Exp(\beta H) = H \star \Exp(\beta H)
\end{equation}
with initial condition $\Exp(0) = 1$. The classical limit of $\Exp(H)$
is the ordinary exponential $\exp(H_0)$. The KMS condition for inverse
temperature $\beta$ and Hamiltonian $H$ for a
$\mathbb{C}[[\hbar]]$-linear functional as formulated in
\cite{basart.flato.lichnerowicz.sternheimer:1984a,
  basart.lichnerowicz:1985a} in the context of deformation
quantization, leads to the following result: up to normalization the
KMS functional is uniquely determined and explicitly given by
\begin{equation}
    \label{eq:KMSFunktional}
    \mu_\kms(f) = \tr(\Exp(-\beta H) \star f)
    \quad
    \textrm{for}
    \quad
    f \in C^\infty_0(M)[[\hbar]]
\end{equation}
see \cite{bordemann.roemer.waldmann:1998a} for the proof and
\cite[Sect.~7.1.4]{waldmann:2007a} for more details on KMS functionals. In
particular, we note that \eqref{eq:KMSFunktional} is a \emph{positive}
functional.
\begin{remark}
    \label{remark:KMSNormierbar}
    Depending on the Hamiltonian $H$, $\mu_\kms$ may or may not be
    normalizable. Whenever the Hamiltonian used permits a
    normalization by rendering the integrations in $\mu_\kms(1)$
    well-defined, we will denote $\frac{1}{\mu_\kms(1)}\mu_\kms$ by
    $\omega_\kms$ and call it a \emph{KMS state}.
\end{remark}

Before entering the particular example again, we note that in the
symplectic case the open quantum time evolution with respect to a KMS
functional is necessarily completely positive. This will follow at
once from this proposition:
\begin{proposition}
    \label{proposition:KMSCPAndContinuous}
    Let the system be an arbitrary Poisson manifold and let the bath be
    symplectic. Given the KMS functional $\mu_\kms$
    with respect to an arbitrary $H_\Bath \in C^\infty(\Bad)[[\hbar]]$ and inverse
    temperature $\beta$, the map $\id \tensor
    \mu_\kms:\rk{\ckomp{\Sys\times\Bad}[[\hbar]], \star}
    \longrightarrow\rk{\ckomp{\Sys}[[\hbar]], \starSys}$ is completely
    positive.
\end{proposition}
\begin{proof}
    We choose a map $\PresSquares{S}$ for the bath whose existence is
    guaranteed by Theorem~\ref{theorem:ExistPosDef}. In the proof of
    Theorem~\ref{theorem:SomeQuantizedOpenIsCP} we have seen that for
    every positive $\mathbb{C}[[\hbar]]$-linear functional $\mu:
    M_n(C^\infty_0(\Sys)[[\hbar]]) \longrightarrow \mathbb{C}[[\hbar]]$ the
    combined map
    \[
    \mu \tensor \PresSquares{S}:
    \left(M_n(C^\infty_0(\Sys \times \Bad)[[\hbar]]), \star \right)
    \longrightarrow
    \left(C^\infty_0(\Bad)[[\hbar]], \cdot\right)
    \]
    is positive. It follows that for $F \in M_n(C^\infty_0(\Sys \times
    \Bad)[[\hbar]])$ the function $(\mu \tensor \PresSquares{S})(F^*
    \star F)$ is at every point $\xb \in \Bad$ either a formal series
    with positive lowest order term or zero.  To avoid trivialities,
    assume that $(\mu \tensor \PresSquares{S})(F^* \star F)$ is not
    identically zero.  Let $r_0$ be the minimal exponent with $(\mu
    \tensor \PresSquares{S})(F^* \star F) = \hbar^{r_0} a_{r_0} +
    \cdots$ and $a_{r_0} \ge 0$ not identically zero. By continuity, there
    is an open subset $U \subseteq \Bad$ with $a_{r_0}(\xb) > 0$ for
    $\xb \in U$. But this implies that $(\mu_\kms \circ
    \PresSquares{S}^{-1}) \circ (\mu \tensor \PresSquares{S})(F^*
    \star F) = \hbar^{r_0} b_{r_0} + \cdots$ with $b_{r_0} > 0$ since
    the zeroth order of $\PresSquares{S}$ is the identity and the
    zeroth order of $\mu_\kms$ is the integration over \emph{all} of
    $\Bad$. Since $\mu$ is arbitrary and using
    \[
    (\mu_\kms \circ \PresSquares{S}^{-1})
    \circ
    (\mu \tensor \PresSquares{S})
    =
    \mu \circ (\id \tensor \mu_\kms),
    \]
    this shows that $(\id \tensor \mu_\kms) (F^* \star F)$ is a
    positive algebra element in $M_n(C^\infty_0(\Sys)[[\hbar]])$ with
    respect to $\starSys$.
\end{proof}

Back to our specific example, we consider the harmonic oscillator as
the Hamiltonian $H_\Bath \in C^\infty(\mathbb{R}^2)$ and the Weyl-Moyal
star product $\star_\Bath$ as before. In this case, the star
exponential of $H_\Bath$ has been computed explicitly by
\cite{bayen.et.al:1978a}. One has
\begin{equation}
    \label{eq:StarExpBAFFL}
    \Exp(-\beta H_\Bath^{\phantom{*}})
    =
    \frac{1}{\cosh\rk{\frac{\hbar\beta\nu}{2}}}
    \exp \rk{
      -\frac{2H_\Bath^{\phantom{*}}}{\hbar\nu}
      \tanh\rk{\frac{\hbar\beta\nu}{2}}
    }
\end{equation}
for $\beta > 0$ and $\nu > 0$, which is a well-defined formal power
series in $\hbar$. Note that in \cite{bayen.et.al:1978a} the
exponential $\Exp(\frac{\I t}{\hbar} H)$ requires a convergent
setting due to the $\hbar$ in the denominator. In our case, the
situation is much simpler. In fact, differentiating
\eqref{eq:StarExpBAFFL} with respect to $\beta$ gives the defining
differential equation~\eqref{eq:ExpODE} right away.

As in the textbooks on statistical mechanics, we can now calculate the
partition function $Z$ as the normalization factor of the KMS state on
the bath by formally calculating Gaussian integrals.
\begin{proposition}
    \label{proposition:PartitionFunction}
    The normalization factor $\mu_\kms(1)$ is explicitly given by
    \begin{equation}
        \label{eq:AlmostThePartitionFunction}
        \mu_\kms(1)
        = 2 \pi \hbar 
        \frac{\exp\rk{-\frac{\hbar\beta\nu}{2}}}
        {1-\exp\rk{-\hbar \beta\nu}}
        \in \mathbb{R}[[\hbar]].
    \end{equation}
    The partition function is the formal Laurent series
    \begin{equation}
        \label{eq:TheRealZ}
        Z =
        \frac{\exp\rk{-\frac{\hbar\beta\nu}{2}}}
        {1-\exp\rk{-\hbar \beta\nu}}
        \in \mathbb{R}((\hbar)).
    \end{equation}
\end{proposition}
The crucial point is that $\mu_\kms(1)$ has a well-defined classical
limit while $Z$ has a simple pole at $\hbar = 0$. Therefore, we can
use this normalization factor to obtain the well-defined KMS state
\begin{equation}
    \label{eq:KMSNormalized}
    \omega_\kms(f)
    =
    \frac{1}{2\pi\hbar Z}
    \int \Exp(-\beta H_\Bath^{\phantom{*}}) \star_\Bath f
    \dbath
\end{equation}
for $f \in C^\infty(\Bad)[[\hbar]]$ such that the integral
\eqref{eq:KMSNormalized} is convergent order by order in $\hbar$. Note
that the inverse of $2\pi\hbar Z$ is again a well-defined formal power
series.

As for the $\delta$-functional, we shall now compute the open quantum
time evolution of the observables $\qs$, $\ps$, and $H_{\System}^{\phantom{*}}$ also
with respect to the KMS state $\omega_\kms$. To this end, we need the
expectation values of $\qb$, $\pb$, $q_\Bath^2$, $p_\Bath^2$, and $\qb\pb$ in
order to evaluate \eqref{eq:TotalTimeEvolutionOfqs},
\eqref{eq:TotalTimeEvolutionOfps}, and
\eqref{eq:TotalTimeEvolutionOfHSystem}.
\begin{lemma}
    \label{lemma:KMSExpectationValuesOfStuff}
    One has the following expectation values
    \begin{equation}
        \label{eq:KMSqbpbpAlleNull}
        \omega_\kms(\qb) = \omega_\kms(\pb) = \omega_\kms(\qb\pb) = 0,
    \end{equation}
    \begin{equation}
        \label{eq:omegaKMSQuadratic}
        \omega_\kms(q_\Bath^2)
        =
        \frac{3\hbar}{2m\nu\tanh\left(\frac{\hbar\beta\nu}{2}\right)},
        \quad
        \textrm{and}
        \quad
        \omega_\kms(p_\Bath^2)
        =
        \frac{3m\nu\hbar}{2\tanh\left(\frac{\hbar\beta\nu}{2}\right)},
    \end{equation}
    which are formal power series in $\C[[\hbar]]$.
\end{lemma}
\begin{proof}
    This is of course textbook knowledge. Nevertheless, we sketch the
    computation in order to illustrate the star product formalism
    used. The first observation is that the trace
    functional $\tr$ for the Weyl-Moyal star product has the
    remarkable feature
    \[
    \tr(f \star g) = \tr(fg),
    \]
    see e.g.\ \cite[Ex.~6.3.33]{waldmann:2007a}. Strictly speaking, one
    of the functions has to have compact support. However, if one is
    the Gaussian $\Exp(-\beta H_\Bath)$ then the rapid decay allows to
    perform the integrations by parts also for observables like
    polynomials. Thus we can use this feature to simplify
    $\omega_\kms(f)$ considerably for the above observables. Since
    $\Exp(- \beta H_\Bath)$ is just a Gaussian we are left with
    the well-known computation of some Gaussian integrals.
\end{proof}

Using these expectation values, we can apply the general
formulas~\eqref{eq:TotalTimeEvolutionOfqs},
\eqref{eq:TotalTimeEvolutionOfps}, and
\eqref{eq:TotalTimeEvolutionOfHSystem} and substitute there the
observables $\qb$, $\pb$, $\qb\pb$, $q_\Bath^2$, and $p_\Bath^2$ by
their expectation values with respect to $\omega_\kms$. This then
gives the open time evolutions of $\qs$, $\ps$ and
$H_{\System}^{\phantom{*}}$. Remarkably, many terms disappear thanks
to the vanishing of \eqref{eq:KMSqbpbpAlleNull}. In detail, we have
\begin{align}
    \label{eq:KMSOpenForqs}
    \sze{A}{t}^{\omega_\kms} \qs
    &=
    \frac{1}{2}
    \left(
        \cos(\nu t) + \cos({\nu_\kappa}t)
    \right) \qs
    + \frac{1}{2}
    \left(
        \frac{\sin(\nu t)}{m\nu} +
        \frac{\sin({\nu_\kappa}t)}{m{\nu_\kappa}}
    \right) \ps,\\
    \label{eq:KMSOpenForps}
    \sze{A}{t}^{\omega_\kms} \ps
    &=
    - \frac{m}{2}
    \left(
        \nu\sin(\nu t) + {\nu_\kappa}\sin({\nu_\kappa}t)
    \right) \qs
    +
    \frac{1}{2}
    \left(
        \cos(\nu t)+\cos({\nu_\kappa}t)
    \right) \ps,\\
    \sze{A}{t}^{\omega_\kms} H^{\phantom{*}}_\System
    &=
    \rk{\frac{m}{8}(\nu\sin(\nu t)+\nu_\kappa\sin({\nu_\kappa}t))^2
      +\frac{m\nu^2}{8}(\cos(\nu t)+\cos({\nu_\kappa}t))^2}q_{\System}^2
    \nonumber \\
    &\quad+
    \rk{\frac{1}{8m}(\cos(\nu t)+\cos({\nu_\kappa}t))^2
      +\frac{m\nu^2}{8}\rk{\frac{\sin(\nu t)}{m\nu}
        +\frac{\sin({\nu_\kappa}t)}{m{\nu_\kappa}}}^2}p_{\System}^2
    \nonumber \\
    &\quad+
    \bigg(-\frac{1}{4}(\nu\sin(\nu t)
    +{\nu_\kappa}\sin({\nu_\kappa}t))
    (\cos(\nu t) +\cos({\nu_\kappa}t))
    \nonumber \\
    &\qquad+
    \frac{m\nu^2}{4}\bigg(\frac{\sin(\nu t)}{m\nu}
    +\frac{\sin({\nu_\kappa}t)}{m{\nu_\kappa}}\bigg)
    (\cos(\nu t)+\cos({\nu_\kappa}t))\bigg)\qs\ps
    \nonumber \\
    &\quad+
    \bigg(\frac{1}{\nu}(\nu\sin(\nu t)-{\nu_\kappa}\sin({\nu_\kappa}t))^2
      +2\nu(\cos(\nu t)-\cos({\nu_\kappa}t))^2
      \nonumber \\
    &\qquad+\nu(\sin(\nu
    t)-\frac{\nu}{{\nu_\kappa}}\sin({\nu_\kappa}t))^2\bigg)
    \frac{3\hbar}{16\tanh\left(
    \frac{\hbar\beta\nu}{2}
    \right)}.
    \label{eq:KMSOpenForHs}
\end{align}

%
%

\begin{footnotesize}
    \renewcommand{\arraystretch}{0.5} 

\end{footnotesize}

\end{document}